\documentclass[a4paper,american,reqno]{amsart}

\usepackage{babel}
\usepackage[utf8]{inputenc}
\usepackage[T1]{fontenc}
\usepackage{lmodern}
\usepackage[binary-units=true]{siunitx}
\usepackage{csquotes}
\usepackage{todonotes}
\usepackage[normalem]{ulem}
\usepackage{amssymb}
\usepackage{xstring}
\usepackage{etoolbox}
\usepackage{mathtools}
\usepackage{tikz}
\usepackage{nicefrac}
\usepackage{xspace}
\usepackage{enumitem}
\usepackage{stmaryrd}
\usepackage{algorithm,algpseudocode}
\usepackage[style = authoryear-comp,
            maxbibnames = 100,
            maxcitenames = 2,
            giveninits = true,
            uniquename = init,
            isbn = false,
            backend = bibtex]{biblatex}
\usepackage[colorlinks,
            citecolor=blue,
            urlcolor=blue,
            linkcolor=blue]{hyperref}
\usepackage{cleveref}
\usepackage{aligned-overset}

\makeatletter
\patchcmd{\@settitle}{\uppercasenonmath\@title}{\scshape\large}{}{}
\patchcmd{\@setauthors}{\MakeUppercase}{\scshape\normalsize}{}{}
\makeatother

\vfuzz \hfuzz
\makeatletter
\@namedef{subjclassname@2020}{%
  \textup{2020} Mathematics Subject Classification}
\makeatother

\newif\ifjournal\journalfalse
\usetikzlibrary{patterns}
\usetikzlibrary{calc}

\newlist{thmparts}{enumerate}{1}
\setlist[thmparts]{
	label=\alph*)
}

\makeatletter
\apptocmd{\cref@getref}{\xdef\@lastusedlabel{#1}}{}{error}
\makeatother

\Crefformat{thmpart}{%
	\StrCount{\@lastusedlabel}{:}[\LastColonPos]%
	\StrBefore[\LastColonPos]{\@lastusedlabel}{:}[\ThmLabel]%
	\nameCref{\ThmLabel}~#2\ref*{\ThmLabel}#1#3%
}
\Crefmultiformat{thmpart}{%
	\StrCount{\@lastusedlabel}{:}[\LastColonPos]%
	\StrBefore[\LastColonPos]{\@lastusedlabel}{:}[\ThmLabel]%
	\nameCref{\ThmLabel}~#2\ref*{\ThmLabel}#1#3%
}{ and~#2#1#3}{, #2#1#3}{ and~#2#1#3}

\Crefname{assumption}{Assumption}{Assumption}
\Crefname{proposition}{Proposition}{Propositions}
\Crefname{theorem}{Theorem}{Theorems}
\Crefname{lemma}{Lemma}{Lemmas}

\newcommand{\jscom}[2][]{
	\ifthenelse{\equal{#1}{journal}}{
	\ifjournal\todo[author=JS, color=yellow!50, size=\small]{#2}\fi}{\todo[author=JS, color=yellow!50, size=\small]{#2}}
	}

\newif\ifshowrevisions\showrevisionsfalse
\colorlet{jsrevcolor}{green!50!black}
\colorlet{adrevcolor}{orange!50!black}
\colorlet{msrevcolor}{orange}
\colorlet{revJScolor}{blue}
\colorlet{jsdelcolor}{red}
\makeatletter
\NewDocumentCommand{\jsdel@cite}{m o o m}{%
  \mbox{\IfNoValueTF{#2}{#1{#4}}{\IfNoValueTF{#3}{#1[#2]{#4}}{#1[#2][#3]{#4}}}}}
\newcommand{\jsdel@boxrefs}{%
  \let\jsdel@ref\ref \def\ref##1{\mbox{\jsdel@ref{##1}}}%
  \let\jsdel@eqref\eqref \def\eqref##1{\mbox{\jsdel@eqref{##1}}}%
  \let\jsdel@cref\cref \def\cref##1{\mbox{\jsdel@cref{##1}}}%
  \let\jsdel@Cref\Cref \def\Cref##1{\mbox{\jsdel@Cref{##1}}}%
  \let\jsdel@@cite\cite \def\cite{\jsdel@cite\jsdel@@cite}%
  \let\jsdel@textcite\textcite \def\textcite{\jsdel@cite\jsdel@textcite}%
  \let\jsdel@parencite\parencite \def\parencite{\jsdel@cite\jsdel@parencite}}

\NewDocumentCommand{\jsdel}{+m}{%
  \ifshowrevisions
    \ifmmode
      {\color{jsdelcolor}\text{\sout{$#1$}}}%
    \else
      {\color{jsdelcolor}\jsdel@boxrefs\sout{#1}}%
    \fi
  \fi}
\makeatother

\DeclareMathOperator*{\argmin}{arg\,min}
\DeclareMathOperator*{\argmax}{arg\,max}
\newcommand{\inte}{\mathrm{int}}
\newcommand{\con}{\mathrm{con}}

\newcommand{\conc}{\mathrm{conc}}
\newcommand{\Proj}{\mathrm{Proj}}

\newcommand{\proxy}{\eta}
\newcommand{\piproxy}{\xi}

\newcommand{\VBR}{\Phi}

\newcommand{\BR}{\mathrm{BR}}
\newcommand{\norm}[1]{\lVert #1 \rVert}
\newcommand{\abs}[1]{\lvert#1\rvert}
\newcommand{\extray}{r}
\newcommand{\mappro}{\alpha}
\newcommand{\addappro}{\beta}

\newcommand{\mydim}{d}
\newcommand{\Lip}{L}

\newcommand{\R}{\mathbb{R}}
\newcommand{\Z}{\mathbb{Z}}
\newcommand{\N}{\mathbb{N}}

\newcommand{\FeasRt}{F_t}
\NewDocumentCommand{\freeset}{O{}}{S^{#1}}
\newcommand{\powset}{\mathcal{P}}
\newcommand{\BRcom}{\mathcal{BR}}

\newcommand{\Nesapp}{\mathcal{E}_{(\mappro,\addappro)}}
\newcommand{\Nespiapp}{\mathcal{E}^{\VBR,\pi}_{(\mappro,\addappro)}}

\NewDocumentCommand{\FeasStrats}{O{}}{W_{#1}}
\NewDocumentCommand{\FeasStratsInt}{O{}}{W_{#1}^\inte}
\NewDocumentCommand{\FeasStratsCon}{O{}}{W_{#1}^\con}
\NewDocumentCommand{\FeasStratsRel}{O{}}{\hat{W}_{#1}}
\newcommand{\Rall}{\prod_{j \in N}\R^{k_j+l_j}}
\newcommand{\Rallmini}{\prod_{j \neq i}\R^{k_j+l_j}}

\newcommand{\cut}[1][]{ANE-cut#1\xspace}
\NewDocumentCommand{\acut}{O{ }}{an ANE-cut#1}
\NewDocumentCommand{\Cplus}{O{}O{ }}{$\appro_{#1}^>$-witness#2}
\newcommand{\Nefreeset}[1][]{$(\mappro,\addappro)$-NE-free set#1\xspace}
\NewDocumentCommand{\minapprox}{O{ }}{simple binary search method#1}
\newcommand{\reusetreesearch}{\textsf{singleTree+Cuts}}
\newcommand{\reusewithoutcuts}{\textsf{singleTree}}
\newcommand{\reusetreeandallcuts}{\textsf{singleTree+allCuts}}
\NewDocumentCommand{\eqtuple}{O{ }}{equilibrium tuple#1}
\NewDocumentCommand{\BinBC}{O{ }}{adaptive B\&C method#1}
\NewDocumentCommand{\aBinBC}{O{ }}{an adaptive B\&C method#1}

\AddToHook{cmd/appendix/before}{\crefalias{section}{appendix}}

\newcommand{\AddThm}[2]{
	\newtheorem{#1}[thmcntr]{#2}
	\AddToHook{env/#1/begin}{\crefalias{thmcntr}{#1}}
}

\theoremstyle{definition}
\AddThm{definition}{Definition}
\AddThm{example}{Example}
\AddThm{proposition}{Proposition}
\AddThm{lemma}{Lemma}
\AddThm{theorem}{Theorem}
\AddThm{conjecture}{Conjecture}
\AddThm{question}{Question}
\AddThm{corollary}{Corollary}
\AddThm{claim}{Claim}
\AddThm{assumption}{Assumption}
\AddThm{remark}{Remark}
\AddThm{invariant}{Invariant}
\AddThm{subclaim}{Subclaim}

\newcommand{\wrt}{w.r.t.\ }

\NewDocumentCommand{\NeInt}{O{\ }}{TODO{#1}}

\newcommand{\defset}[3][\defsep]{\set{#2#1#3}}
\newcommand{\Defset}[3][\defsep]{\Set{#2#1#3}}
\newcommand{\set}[1]{\{#1\}}

\newcommand{\st}{\text{s.t.}}
\newcommand{\define}{\mathrel{{\mathop:}{=}}}

\newcommand{\nonNegInts}{\mathbb{Z}_{\geq 0}}

\bibliography{approximate-gne-via-bnc}

\begin{document}

\title[Branch-and-Cut for Computing Approximate Equilibria of Nash Games]%
{Branch-and-Cut for Computing Approximate Equilibria\\
  of Mixed-Integer Generalized Nash Games}

\author[A. Duguet, T. Harks, M. Schmidt, J. Schwarz]%
{Aloïs Duguet, Tobias Harks, Martin Schmidt, Julian Schwarz}

\address[A. Duguet, M. Schmidt]{%
  Trier University,
  Department of Mathematics,
  Universitätsring 15,
  54296 Trier,
  Germany}
\email{duguet@uni-trier.de}
\email{martin.schmidt@uni-trier.de}

\address[T. Harks, J. Schwarz]{%
  University of Passau,
  Faculty of Computer Science and Mathematics,
  94032 Passau,
  Germany}
\email{tobias.harks@uni-passau.de}
\email{julian.schwarz@uni-passau.de}

\date{\today}

\begin{abstract}
  Generalized Nash equilibrium problems with mixed-integer
variables constitute an important class of games in which
each player solves a mixed-integer optimization problem,
where both the objective and the feasible set are parameterized by the
rivals' strategies.
However, such games are known for failing to admit exact equilibria
and the assumption of all players being able to solve nonconvex problems to
global optimality is questionable.
This motivates the study of approximate equilibria.
We consider an approximation concept that incorporates
both multiplicative and additive relaxations of optimality.
We propose a branch-and-cut (B\&C) method that, upon termination, computes such
approximate equilibria or proves their non-existence.
For this, we adopt intersection cuts and show their existence under the assumption of
affine integer-valued constraints
and that each player’s cost function satisfies one of three
conditions: (i) convexity and differentiability in the entire strategy profile, (ii)
concavity in the entire strategy profile and an affine dependence on the rivals’
strategies, or, (iii) convexity in the rivals' strategies and
Lipschitz continuity in the entire strategy profile.
For the special case of standard Nash equilibrium problems, we
introduce an alternative type of cut and show that the method
terminates finitely, provided that each player has only finitely many
distinct best-response sets.
Based on these cuts, we further discuss a fallback strategy that
guarantees finite termination for generalized Nash equilibrium problems
under similar assumptions.
Finally, based on the B\&C method, we introduce
a single-tree binary-search method to compute best-approximate
equilibria.
We implemented these methods and present numerical results
for a class of mixed-integer flow games and for generalized Nash
equilibrium problems with concave quadratic cost functions.


\end{abstract}

\keywords{Nash equilibrium problems,
Generalized Nash equilibrium problems,
Mixed-integer games,
Approximate equilibria,
Branch-and-cut%
%
%
}
\subjclass[2020]{90C11, 
90C57, 
91-08
%
}

\maketitle

\section{Introduction}
\label{sec:introduction}

Generalized Nash equilibrium problems (GNEPs) arise in various
domains including market games in economics \parencite{ArrowDebreu},
communication networks \parencite{Kelly98},
transportation systems \parencite{Beckmann56},
and electricity markets \parencite{Anderson13}.
Since the seminal works of \textcite{ArrowDebreu}, significant
progress has been made in understanding the existence and computation
of generalized Nash equilibria (GNEs).
One key assumption for ensuring the existence of equilibria is the
convexity of the data of the game.
The problem, however, is that many timely and important applications
of GNEPs contain substantial nonconvexities.
In this paper, our focus lies on nonconvexities with respect to the
  players' strategy sets, which typically stem from integrality
  restrictions on (some of) the players' decisions.
One prominent example are power markets on which both electricity
producers and consumers act.
The mixed-integer nature of the power producers' models render the
resulting GNEPs highly complex and lead to the invalidity of classic
existence theorems; see, e.g.,
\textcite{Liberopoulos_Andrianesis:2016,guo2021copositive} and the
references therein.
For instance, based on the complete characterization of the existence
of equilibria in \textcite{HarksSchwarzPricing},
\textcite{Gruebel_et_al:2023a} prove the non-existence of equilibria
for many power as well as gas market instances.

Consequently, algorithms designed to compute GNEs may fail to
terminate and often cannot provide any conclusive information about
the equilibrium problem under consideration.
Nevertheless, since the underlying real-world applications are
of utmost importance, there is the clear need for alternative solution
concepts that are less restrictive than classic Nash equilibria.
Another aspect that goes along with the nonconvexities in many
applications is that practical instances are usually not solved to
global optimality anyways---a classic example of bounded rationality
in theoretical economics \parencite{Simon:1972,Rubinstein:1998}, which
arises when the decision-making frameworks of the agents are too complex.
To address both of these fundamental challenges, we
study the concept of approximate mixed-integer GNEs, for which we (i)
introduce a new and more general notion of approximation, (ii) provide
the first branch-and-cut (B\&C) method to solve these problems, (iii)
develop further algorithmic enhancements to compute best-approximate
GNEs, and (iv) present a small numerical study that shows the
applicability of our approach.

\subsection{Our Contributions}

We now outline our approximate equilibrium concept and the B\&C
framework in greater detail in the following.

\vspace*{0.5em}
\paragraph{\textbf{$(\mappro,\addappro)$-Nash Equilibria}}
\label{sec:alpha-beta-nash}

We study an equilibrium concept that incorporates a multiplicative as
well as an additive relaxation of optimality.
More precisely, we say that a strategy profile $x$ is an
$(\mappro,\addappro)$-Nash equilibrium ($(\mappro,\addappro)$-NE) if
\begin{align*}
  \pi_i(x) \leq \mappro_i \pi_i(y_i,x_{-i}) + \addappro_i
  \quad \text{for all} \quad
  y_i \in X_i(x_{-i})
\end{align*}
holds for all players~$i$.
Here, $\pi_i$ denotes the cost function of player $i$, $(y_i,x_{-i})$
denotes the strategy profile in which player $i$ unilaterally deviates
to a strategy~$y_i$, which is contained in her set of feasible
strategies~$X_i(x_{-i})$ under the rivals' strategies~$x_{-i}$.
The above condition then requires that no player can improve her costs by
unilaterally changing her strategy up to a player-specific
multiplicative factor~$\mappro_i \geq 1$ together with an additive
value~$\addappro_i \geq 0$.
Developing our framework on the basis of this approximate equilibrium
concept allows to choose the approximation values in an
instance-specific way and even to consider a Pareto-frontier of
minimum $(\mappro,\addappro)$ values such that a corresponding
approximate NE exists.
This is highly beneficial as the appropriate notion of approximation
crucially depends on each specific instance, e.g., on the range of
values of the cost functions.
We refer to \textcite{Daskalakis13} for a detailed discussion of pros
and cons of multiplicative and additive approximations of equilibria.

\vspace*{0.5em}
\paragraph{\textbf{A B\&C Framework}}
\label{sec:bc-framework}

We propose a novel B\&C framework to compute $(\mappro,\addappro)$-NEs
or to prove non-existence.
To this end, we reformulate the problem of finding an
$(\mappro,\addappro)$-NE as an optimization
problem in which we minimize an auxiliary variable~$\lambda$
that represents via proper constraints the maximum over all players'
approximate regrets, i.e., the deviation between their current costs
and a respective best-response value scaled by $\mappro_i$
and additively increased by $\addappro_i$.
Inspired by techniques from bilevel optimization, we relax this problem by
relaxing the integrality constraints for the strategy profiles as
well as introducing for every player $i$ two
proxy-variables~$\proxy_i$ and~$\piproxy_i$ and substituting those for
the respective best-response value and costs.
This relaxed problem is then embedded into a B\&C method:
We branch on fractional variables and via suitably constructed cuts in
the space of strategy profiles, maximal regret, and proxy
variables $(x,\lambda,\proxy,\piproxy)$, the approximation quality of
the proxy-variables is iteratively improved until an
$(\mappro,\addappro)$-NE is found or proven that none exists (in the
current node); see Section~\ref{sec:algorithm}.
Further below, we discuss an adaptive extension of this method to
approximate the infimum of one multiplicative approximation
parameter, with all remaining approximation parameters fixed, subject
to the existence of a corresponding approximate equilibrium.

\vspace*{0.5em}
\paragraph{\textbf{New Cuts}}
\label{sec:new-cuts}

In this B\&C method, the main challenge lies in constructing cuts that
exclude integer-feasible optimal node solutions
$(x^*,\lambda^*,\proxy^*,\piproxy^*)$ while preserving all
\eqtuple[s], i.e., tuples $(x,\lambda,\proxy,\piproxy)$
where $x$ represents an equilibrium, the regret satisfies $\lambda=
0$, and $\proxy,\piproxy$ correspond to their respective approximation
quantities (best-response value and costs at $x$).
For GNEPs, we employ the theory of intersection cuts (ICs), which
were originally introduced by \textcite{balas1971intersection} in the
context of integer programming.
\begin{figure}
  \begin{tikzpicture}[xscale=1]

  \coordinate (A) at (0,1.25);
  \coordinate (B) at (-0.5,3.5);
  \coordinate (C) at (4,0.75);
  \coordinate (D) at (4.5,3);
  \coordinate (F) at (-0.22,2.26);
  \coordinate (E) at (4.25,1.88);
  \coordinate (x) at (2,3);
  \coordinate (s2) at (4,2);
  \coordinate (s1) at (0,2.33);

  \draw[thick,orange,fill=orange!20] (A) -- (C) -- (E) -- (x) -- (F) -- cycle;

  \draw ($ (A) + 0.5*(A) - 0.5*(B) $) -- (B);
  \draw ($ (A) + 0.4*(A) - 0.4*(C) $) -- ($(C) + 0.4*(C) -0.4*(A)$);
  \draw (D) -- ($(C) + 0.4*(C) -0.4*(D)$);

  \node[circle, fill=blue, inner sep=1.5pt]  at (2,1.5) {};
  \node[circle, fill=blue, inner sep=1.5pt]  at (3,1.75) {};
  \node[circle, fill=blue, inner sep=1.5pt]  at (1,2) {};
  \node[circle, fill=blue, inner sep=1.5pt]  at (0.5,1.2) {};

  \draw[fill = green!20] (x) circle (0.8);

  \node[circle, fill=black, inner sep=1.5pt]  at (2,3) {};
  \node at (2,3.2)   {$(x^*,\lambda^*,\proxy^*,\piproxy^*)$};
  \node (v1) at (5,1.5) {};
  \node (v2) at (-1,2) {};

  \node[circle, fill=black, inner sep=1.5pt]   at (2.71,2.64) {};
  \node[circle, fill=black, inner sep=1.5pt]   at (1.241, 2.747) {};
  \coordinate (s3) at (2.71,2.64);
  \coordinate (s4) at (1.241, 2.747);

  \draw[thick,->] (x) to node[below left] {$\extray_2$} (v1);
  \draw[thick,->] (x) to node[below right] {$\extray_1$} (v2);

  \draw[red,thick,dashed] ($ (s3) + 2*(s3) - 2*(s4) $) -- node[above
  right, pos=0.15] {IC} ($ (s4) + 2*(s4) - 2*(s3) $);

\end{tikzpicture}

  \caption{Sketch of an intersection cut in the context of
    $(\mappro,\addappro)$-NE in mixed-integer GNEPs}
  \label{fig:Intro}
\end{figure}
We illustrate our contribution and the main challenges using
the schematic sketch in Figure~\ref{fig:Intro}.
Broadly speaking, constructing an IC requires two elements:
an \Nefreeset, i.e., a convex set with $(x^*,\lambda^*,\proxy^*,\piproxy^*)$ in its interior (green circle in the figure),
and a cone pointed at $(x^*,\lambda^*,\proxy^*,\piproxy^*)$ that contains all \eqtuple[s].
The way we re-state the  highly nonlinear problem of finding an
$(\mappro,\addappro)$-NE  as an
optimization problem has the key feature that the node problems in our
B\&C framework are linear programs (with the orange area in the figure
representing the feasible set). This ensures that the optimal node solution
is attained at a vertex of the polyhedral feasible set (as depicted in the figure) which
allows to use the LP basis cone at the
solution~$(x^*,\lambda^*,\proxy^*,\piproxy^*)$ of the node that we need to cut (the cone induced by the extreme rays $r_1$ and $r_2$ in the figure),
guaranteeing the existence of a computational tractable cone required for the IC.
Hence, within our framework, the existence of an IC boils down to the existence of an \Nefreeset
which we are able to construct (see Section~\ref{sec:cuts-gne}) under the following
assumptions: We have integer-valued affine constraints and
  each player's cost function
satisfies one of the following three
conditions: (i) convexity and differentiability in the entire strategy profile, (ii)
concavity in the entire strategy profile and an affine dependence on the rivals’
strategies, or, (iii) convexity in the rivals' strategies and
Lipschitz continuity in the entire strategy profile.

Moreover, for the special case of standard NEPs, we consider in
\Cref{sec:cuts-ne} a different type of cut, which is not based on ICs
and whose validity does not require any assumptions on the cost or constraint functions.
For this (possibly nonlinear) cut, we prove finite termination of our
algorithm under the assumption that the set of best-response sets for
each player is finite.
This holds true in particular for the important special cases of
(i) players' constraints being affine and their cost functions
being concave in their own continuous
strategies and (ii) the players' cost function only depending on their
own strategy and the rivals integer strategy components.
For GNEPs, in contrast, ICs may become arbitrarily shallow so that
finite termination is not guaranteed.
In this regard, we argue in \Cref{rem: Fallback} that the fallback
strategy of \textcite{DSHS25} carries over to our setting:
Once an integer profile recurs too often in the node solutions, we fix
it by branching and switch to the above cuts in the resulting node.
These cuts are valid there if the players' strategy sets do not depend
on the rivals' continuous strategies.
Under the additional assumption that each player has only finitely many
distinct best-response sets for every fixed integer profile, this
guarantees finite termination for GNEPs as well.

\vspace*{0.5em}
\paragraph{\textbf{Computing Best-Approximate Equilibria}}

We further provide in \Cref{sec:best-approximate} a method based on
our B\&C algorithm that, for any player $i$ and given approximation values
$\mappro_j,j\neq i$, and $\addappro_j,j\in N$, computes (up to a given tolerance) the infimum of all values $\mappro_i \geq 1$
such that a $(\mappro,\addappro)$-NE exists.
This algorithm is based on a binary search over the value
$\mappro$.
Remarkably, we can enhance this basic procedure substantially by
reusing information gained during the execution of the B\&C in former
binary search iterations.
More precisely, whenever an approximate equilibrium \wrt $\mappro_i$
is found and the value of $\mappro_i$ is decreased to
$\mappro'_i < \mappro_i$, we show that we can reuse the entire tree
structure (including the derived cuts) of the B\&C method for
tackling~$\mappro'_i$ as previously pruned nodes that
were cut off do not contain approximate equilibria \wrt $\mappro_i$ and
thus, in particular, not for any smaller value.
Exploiting this, we can turn the basic binary search, which is a
multi-tree method, into a more effective single-tree algorithm.

\vspace*{0.5em}
\paragraph{\textbf{Numerical Case Study}}

Finally, we implemented our B\&C method and the above procedure to
compute best-approximate NEs.
Here, we restrict ourselves to the
special situation that every player has the same multiplicative
approximation guarantee.
Preliminary results for NEPs arising in mixed-integer flow games
and for GNEPs with concave quadratic cost functions are
presented in \Cref{sec:numerical-results}.

\subsection{Related Work}

Over the recent decades, there has been growing interest in the
understanding and computation of equilibria in nonconvex games.
For the particularly challenging class of nonconvex GNEPs,
however, only a handful of studies exist to date. Sagratella
was the first to tackle this problem in the context of
Cournot oligopoly models with mixed-integer
quantities \parencite{sagratella2017computing} and generalized
mixed-integer potential games \parencite{Sagratella2017AlgorithmsFG},
demonstrating that a best-response algorithm converges in
finitely many steps to an additive $\varepsilon$-approximate
equilibrium, i.e., a $((1,\ldots,1),(\varepsilon,\ldots,\varepsilon))$-NE,
for any given approximation value $\varepsilon>0$.
For similar approaches based on best-response methods
for mixed-integer GNEPs; see \textcite{Fabiani_Grammatico:2020}
and \textcite{Fabiani_et_al:2022}.
In \textcite{Sagratella19}, the author considers mixed-integer
GNEPs with affine coupling constraints and proposes a branch-and-bound
(B\&B) method under the strong assumption of an existing and
computationally  tractable merit function.
Moreover, a branch-and-prune (B\&P) method is developed that exploits
the idea of dominance of strategies for pruning.
In this regard, let us remark that the pruning steps also employ
``cuts''.
However, rather than tightening a relaxation as in our paper,
these cuts are used to prune branches of the B\&P tree.
\textcite{Harks24} tackle nonconvex GNEPs via
a convexification technique, associating to every GNEP a
corresponding set of convexified instances with the same set
of originally feasible equilibria. They then introduce the class
of quasi-linear GNEPs and show how their convexification approach can be
used to reformulate the original GNEP as a standard (nonlinear)
optimization problem and provide a numerical study for this class.
Their general approach
is limited in the sense that it relies on deriving a convexification
which itself is known to be computationally difficult.
\textcite{DSHS25}, and also we in this work, circumvent
this problem by offering a  direct computational approach, which is
the first B\&C framework for the computation of exact (pure) equilibria for
general nonconvex GNEPs.
They use the Nikaido--Isoda~(NI) function to reformulate the GNEP and
use ideas from bilevel optimization to set up their B\&C framework.
In contrast to our approach, their reformulation cannot handle
$(\mappro,\addappro)$-NEs due to the aggregative nature of the NI
function used for their node problems.
Moreover, similar to the framework presented here, their approach
relies on the existence of suitable cuts and, in particular, on the
existence of ICs.
In contrast to our framework, they need to require the restrictive assumption
of social costs being concave in order to guarantee that an optimal node solution
is attained at a vertex of the feasible set and to be able to use the corresponding
LP basis cone for the construction of an IC.
Our specific node problem formulation allows us to bypass
such restrictive assumptions in exchange for milder conditions.
In particular, our framework applies to any instance that can be tackled by the approach of \textcite{DSHS25}
if one additionally assumes Lipschitz continuity of the players' costs.
Moreover, their fallback strategy for ensuring finite termination
in the GNEP setting remains applicable in our framework; see
\Cref{rem: Fallback}.

For the simpler case of standard NEPs, so-called
integer programming games (IPGs) are introduced in~\textcite{Koeppe11}
and have been the subject of extensive research;
see, e.g., \textcite{cronert2022equilibrium,KleerS17,PiaFM17} as well as the survey of
\textcite{carvalho2023integer} and references therein.
Note that several papers,
e.g.,~\textcite{carvalho2021cut,carvalho2022computing,cronert2022equilibrium},
tackle the non-existence of (pure) equilibria and the arising
difficulties from the nonconvexity of IPGs by focusing on mixed NEs.
In contrast, we focus on pure NEs but note that mixed NEs of finite games correspond
to pure NEs of their mixed extensions and can therefore also be
handled within our framework.
In the following, we focus on works that compute pure equilibria leveraging
mixed-integer programming techniques.
\textcite{dragotto2023zero} address the computation, enumeration, and
selection of Nash equilibria in IPGs with purely integral strategy
spaces using a cutting-plane algorithm. In contrast to our
B\&C framework, their method solves integer programs at
intermediate steps and does not involve branching on fractional
points.
\textcite{kirst2022branch} propose a B\&B algorithm for
computing the set of all $\varepsilon$-additive approximate equilibria
within a specified error tolerance for IPGs with box-constraints.
By exploiting this special structure, their approach relies on rules
that identify and eliminate regions of the feasible space that cannot
contain any such equilibrium.
Similar pruning arguments, exploiting the dominance of strategies
based on the derivative of the cost function,
were used for standard NEPs with convex constraints in the pure
integer case. Here, \textcite{Sagratella16} proposes a
branching method to compute the entire set of NE, which was
further enhanced by \textcite{schwarze2023branch} via a pruning
procedure.

Note that some of the above
papers~\parencite{dragotto2023zero,kirst2022branch,Sagratella2017AlgorithmsFG}
also consider additive $\varepsilon$-approximate equilibria.
Outside the context of general formulations of mixed-integer (G)NEPs,
approximate equilibria have been studied in various settings.
We refer to~\textcite{Deligkas20continuous} for the case of continuous
games.
Recently, approximate mixed equilibria were studied in
\textcite{Duguet25MNE} for nonconvex cost functions.
In the special case of integer weighted congestion games,
generalizations of the notion of potential games led to approximation
results for the value such that a multiplicative approximate NE exists
\parencite{Caragiannis15,Hansknecht14}.
For market equilibria, other approximation concepts
have been derived. For approximate solution concepts relaxing the
strategy spaces, e.g., the market-clearing condition; see
\textcite{Budish11Shapley,Deligkas24Inapproximability,%
	Guruswami05envyfree,Vazirani11Market}.
Here, the Shapley--Folkman theorem~\parencite{Starr12Book}
constitutes an important tool to prove existence of approximate
equilibria; see, e.g., \textcite{Liu23Shapley} for sum-aggregative games.
Finally, multiplicative notions of approximate market equilibria have
been considered in~\textcite{Codenotti05Market,Garg25Welfare} for
special (concave) cost functions and divisible goods.
Let us note here that approximate equilibria have been used in the algorithmic game
theory community as a relaxation concept in order to obtain polynomial time algorithms
computing them; cf.~\textcite{Rubinstein17}. Note that our general
game formulation already implies hardness of computing approximate
Nash equilibria for any constant factor (and also hardness for the additive approximation variant
in the order of the largest cost/utility value of any player in the game).
To see this, one can simply reduce from the directed two arc-disjoint paths problem via creating a corresponding integer GNEP with a jointly constrained flow polytope and
unit edge capacities.


\section{Problem Statement}
\label{sec:problem-statement}

We consider a non-cooperative and complete-information game~$G$
with players indexed by the set~$N = \set{1, \dotsc,n}$.
Each player~$i \in N$ solves the optimization problem
\begin{equation}
  \label{opt: player}
  \tag{$\mathcal{P}_i(x_{-i})$}
  \begin{split}
    \min_{x_i} \quad & \pi_i(x_i,x_{-i}) \\
    \st  \quad & x_i \in X_i(x_{-i}) \subseteq  \Z^{k_i} \times \R^{l_i},
  \end{split}
\end{equation}
where $x_i$ is the strategy of player~$i$ and $x_{-i}$ denotes the
vector of strategies of all players except player~$i$.
The function $\pi_i:  \prod_{j \in N}  \R^{k_j + l_j} \to \R$ denotes
the cost function of player $i$.
The strategy set $X_i(x_{-i})$ of player~$i$ depends on the rivals'
strategies $x_{-i}$ and is a subset of $\Z^{k_i} \times \R^{l_i}$ for
$k_i,l_i \in \nonNegInts$, i.e., the first $k_i$ strategy components
are integral and the remaining~$l_i$ are continuous variables.
We assume that the strategy sets are of the form
\begin{align*}
  X_i(x_{-i}) \define\Defset{x_i \in \Z^{k_i} \times
  \R^{l_i}}{g_i(x_i,x_{-i}) \leq 0}
\end{align*}
for a function $g_i: \prod_{j \in N}  \R^{k_j+ l_j} \to \R^{m_i}$ and
$m_i \in \nonNegInts$.
We denote by $X(x) \define \prod_{i \in N} X_i(x_{-i})$ the product set
of feasible strategies \wrt $x$ and by
\begin{align*}
  \FeasStrats\define\Defset{x \in \prod_{i \in N} \Z^{k_i} \times
  \R^{l_i}}{x \in X(x)} = \Defset{x \in \prod_{i \in N} \Z^{k_i}
  \times  \R^{l_i}}{g(x) \leq 0}
\end{align*}
the set of feasible strategy profiles, where we abbreviate $g(x)
\define (g_i(x))_{i \in N}$.
We also use its continuous relaxation defined by
\begin{align*}
  \FeasStratsRel\define \Defset{ x \in \prod_{i \in N}\R^{k_i +
  l_i}}{g(x) \leq 0}.
\end{align*}

In order to guarantee that \eqref{opt: player} and our B\&C node
problems (discussed below) admit an optimal solution (if feasible),
we make the following standing assumption.
\begin{assumption}\label{ass:GeneralAss} \hfill
  \begin{enumerate}[label=(\roman*)]
  \item $\FeasStrats$ is non-empty and
    $\FeasStratsRel$  is compact. \label[thmpart]{ass:GeneralAss:Poly}
  \item For every player $i$, her cost function is bounded
    and lower semi-continuous on $\FeasStratsRel$.
  \item For every player~$i$ and every $x \in \FeasStrats$, the
    set~$X_i(x_{-i})$ is compact and $\pi_i(\cdot,x_{-i})$ is lower
    semi-continuous on~$X_i(x_{-i})$.\label[thmpart]{ass:GeneralAss:BR}
  \end{enumerate}
\end{assumption}
\Cref{ass:GeneralAss}\ref{ass:GeneralAss:BR} guarantees that \eqref{opt: player} has an
optimal solution, i.e., that a best response exists, for all
$x \in \FeasStrats$.
This does not follow from the first two conditions of
\Cref{ass:GeneralAss} in general since a unilateral deviation
$(y_i,x_{-i})$ with $y_i \in X_i(x_{-i})$ only needs to satisfy the
constraints of player~$i$ and, thus, may not be contained
in~$\FeasStratsRel$.
However, \Cref{ass:GeneralAss}\ref{ass:GeneralAss:BR} is implied by the first two conditions
for standard NEPs (see \Cref{sec:cuts-ne}) as well as for GNEPs with
shared constraints, i.e., if $g_i = g_j$ holds for all $i,j \in N$,
because the set $\defset{(y_i,x_{-i})}{y_i \in X_i(x_{-i})}$ is a
closed subset of~$\FeasStratsRel$ for all $x \in \FeasStrats$ in these
cases.

We consider now a multiplicative approximation vector
$\mappro=(\mappro_i)_{i\in N} \in [1,\infty)^N$
and an additive approximation vector $\addappro=(\addappro_i)_{i\in N}
\in [0,\infty)^N$.
A  strategy profile~$x^* \in \FeasStrats$ is called an
$(\mappro,\addappro)$-Nash equilibrium
($(\mappro,\addappro)$-NE) if
\begin{equation*}
  \pi_i(x_i^*,x_{-i}^*) \leq \mappro_i \VBR_i(x_{-i}^*) + \addappro_i
  \quad  \text{ for all } i \in N
\end{equation*}
holds, where
\begin{equation*}
  \VBR_i(x_{-i})
  \define \min_{y_i \in X_i(x_{-i})} \ \pi_i(y_i, x_{-i})
\end{equation*}
denotes the best-response value for player $i$ \wrt $x_{-i}$.
We further denote by $\VBR(x) \define (\VBR_i(x_{-i}))_{i\in N}$ the
vector of best-response values and call a vector $(y^*_i)_{i\in N}$
with $y^*_i \in \argmin \Defset{ \pi_i(y_i,x_{-i}^*) }{y_i\in
  X_i(x_{-i}^*)}$ a best-response vector to
$x^*$.
  
\begin{remark} \label{rem:NE_definition}
The definition of approximate Nash equilibrium is meaningful when the best responses $\VBR(x)$ are nonnegative.
If a best-response is nonpositive, one should instead check
\begin{equation*}
  \pi_i(x_i^*,x_{-i}^*) \leq \frac{1}{\mappro_i} \VBR_i(x_{-i}^*) + \addappro_i.
\end{equation*}
For the sake of clarity, in the rest of this work we assume that best-responses are always nonnegative, i.e., $\VBR_i(x_{-i}) \geq 0$ holds for all $x \in \FeasStrats$ and $i \in N$.
Note however that our results remain valid for games where best-responses are always nonpositive by using the alternative approximate Nash equilibrium definition.
In this case, the constraints~\eqref{eq: LambdaBound} in the
reformulation~\eqref{model:R} below and in all node problems have to be
replaced by
\begin{equation*}
  \lambda \geq \mappro_i \piproxy_i - \proxy_i - \mappro_i \addappro_i
  \quad \text{for all } i \in N.
\end{equation*}
Note further that neither of the two definitions is meaningful for
games in which the sign of $\VBR_i(x_{-i})$ depends on~$x_{-i}$.
\end{remark}

We use the following notation throughout the paper.
We denote the set of all $(\mappro,\addappro)$-NE by $\Nesapp
\subset \FeasStrats$.
By $\FeasStrats[i]\define \defset{x_i}{\exists x_{-i} \text{ with }
	(x_i,x_{-i}) \in \FeasStrats}$ we refer to the  projection
of~$\FeasStrats$ to the strategy space of player~$i$ and define analogously 
$\FeasStrats[-i]\define \defset{x_{-i}}{\exists x_{i} \text{ with }
	(x_i,x_{-i}) \in \FeasStrats}$.


\section{The Algorithm}
\label{sec:algorithm}

For given approximation vectors  $(\mappro,\addappro) \in
[1,\infty)^N\times [0,\infty)^N$, we now
derive a branch-and-cut (B\&C) algorithm to compute an
$(\mappro,\addappro)$-NE or prove that none exists for the GNEP
defined in \Cref{sec:problem-statement}.
Using this algorithm as a subroutine allows the study of
Pareto-minimal approximation values $(\mappro,\addappro)$ such that a
corresponding $(\mappro,\addappro)$-NE exists;
cf.~Section~\ref{sec:best-approximate}.

To this end, we
reformulate the problem to find an
$(\mappro,\addappro)$-NE as the optimization problem
\begin{align}
  \tag{$\mathcal{R}$}
  \label{model:R}
    \min_{\lambda \in \R,x \in \FeasStrats,\proxy\in\R^N,\piproxy\in\R^N}\quad
    &\lambda \\
    \st  \quad
    & \lambda \geq \frac{\piproxy_i}{\mappro_i} -\proxy_i
    -\frac{\addappro_i}{\mappro_i}
    \quad \text{ for all }i\in N, \label{eq: LambdaBound}\\
    &\proxy\leq \VBR(x) \label{eq: ProxyBound}\\
    &\piproxy \geq \pi(x).  \label{eq: PiProxyBound}
\end{align}
Here, $\proxy,\piproxy\in \R^N$  can be seen as proxy variables that
approximate the best-response values and costs at a strategy profile $x$, respectively.

\begin{lemma}\label{lem:reformulation}
  A strategy profile~$x \in \Rall$
  is an $(\mappro,\addappro)$-NE if and only if
  there exists a feasible point~$(x,\lambda,\proxy,\piproxy)$ for \eqref{model:R}
  with $\lambda \leq 0$.
  In particular, there exists an $(\mappro,\addappro)$-NE if and only if
  \eqref{model:R} has a feasible point with $\lambda \leq 0$.
\end{lemma}
Note that if \eqref{model:R} has an optimal solution, \eqref{model:R} has a feasible
point with $\lambda \leq 0$ is equivalent to
    the optimal value of \eqref{model:R} being smaller or equal to $0$.
\begin{proof}
  We start by showing the if-direction, i.e., let $(x,\lambda,\proxy,\piproxy)$ be
  feasible for~\eqref{model:R} with $\lambda \leq 0$. Then
   $x \in \FeasStrats$ is a feasible strategy profile
  and  for all $i\in N$, we get
  \begin{align*}
    0 \geq \lambda \overset{\eqref{eq: LambdaBound}}{\geq} \frac{\piproxy_i}{\mappro_i}
      -\proxy_i-\frac{\addappro_i}{\mappro_i}
      \overset{\eqref{eq: ProxyBound},\eqref{eq: PiProxyBound}}{\geq}
      \frac{\pi_i(x)}{\mappro_i} -\VBR_i(x_{-i})-\frac{\addappro_i}{\mappro_i}
    \implies \pi_i(x) \leq \mappro_i \VBR_i(x_{-i}) + \addappro_i.
  \end{align*}
  Hence, $x$ is an $(\mappro,\addappro)$-NE.

  Now, we prove the only-if-direction, i.e., let $x$ be an
  $(\mappro,\addappro)$-NE.
  Then, $x \in \FeasStrats$ and for all $i\in N$, we have
  \begin{align*}
    \pi_i(x) \leq \mappro_i \VBR_i(x_{-i}) + \addappro_i \implies 0 \geq
    \frac{\pi_i(x)}{\mappro_i} -\VBR_i(x_{-i}) -\frac{\addappro_i}{\mappro_i} ,
  \end{align*}
  showing that $(x,0,\VBR(x),\pi(x))$ is feasible for \eqref{model:R}.
\end{proof}

Consequently, we are looking for a global minimizer~$(x,\lambda,\proxy,\piproxy)$ of
\eqref{model:R}.
To this end, we draw inspiration from bilevel optimization and the
so-called continuous high-point relaxation (C-HPR).
That is, we relax~\eqref{model:R} by relaxing the
integrality constraints of the strategy profiles using
$\FeasStratsRel$ instead of $\FeasStrats$.
Moreover, we relax both of the inequality constraints \eqref{eq:
  ProxyBound} and \eqref{eq: PiProxyBound} to obtain
\begin{align}
  \tag{C-HPR}
  \label{model:CpiHPR}
  \min_{x \in \FeasStratsRel, \lambda \in \R,\proxy \in \R^N,\piproxy \in \R^N} \quad
  & \lambda& \\
  \st \quad
  &\lambda \geq \frac{\piproxy_i}{\mappro_i}-\proxy_i -
    \frac{\addappro_i}{\mappro_i}  \quad \text{ for
                                     all }i\in N,\nonumber\\
  &\proxy \leq \proxy^+, \quad \piproxy\geq \piproxy^-, \nonumber
\end{align}
where we use $\proxy^+ \in \R^N$ ($\piproxy^- \in \R^N$) as a finite upper-bound (lower-bound)
vector for $\proxy$ ($\piproxy$) to ensure boundedness of the problem.
For them to be valid, we require that $\proxy^+_i \geq
\VBR_i(x_{-i})$ for all $x \in \FeasStrats$
and $\piproxy^-_i \leq \pi_i(x)$ for all $x \in \FeasStrats$.
Examples for such bounds are
\begin{equation*}
  \proxy_i^+ = \sup \Defset { \pi_i(x) }{x \in
    \FeasStrats }
  \quad \text{and} \quad
  \piproxy_i^- = \inf \Defset { \pi_i(x) }{x \in \FeasStrats}.
\end{equation*}
Note that these values are finite by the assumption that $\pi_i$ is
bounded on~$\FeasStratsRel \supseteq \FeasStrats$.
Moreover, $\proxy_i^+$ is valid since $\VBR_i(x_{-i}) \leq \pi_i(x)$
holds for all $x \in \FeasStrats$.

We aim to embed \eqref{model:CpiHPR} in a B\&C method.
Note that  in a B\&C tree, each node problem is given by the root-node problem
\eqref{model:CpiHPR} together with additional constraints.
These constraints correspond to the branching decisions and cuts
added along the path from the root node to node~$t$.
We denote them by $B_t$ and $C_t$, respectively.
Hence, the problem at node~$t$ can be formulated as
\begin{equation}
  \tag{$\mathcal{R}_t$}
  \label{model:node_t}
  \begin{split}
    \min_{x,\lambda,\proxy,\piproxy} \quad
    & \lambda
    \\
    \st \quad & \lambda \geq \frac{\piproxy_i}{\mappro_i} - \proxy_i
    -\frac{\addappro_i}{\mappro_i}
    \quad \text{ for all }i\in N,\\
    & \proxy \leq \proxy^+, \quad  \piproxy \geq \piproxy^-,
    \\
    & (x,\lambda,\proxy,\piproxy) \in (\FeasStratsRel \times \R\times
    \R^N\times \R^N)
    \cap B_t \cap C_t.
  \end{split}
\end{equation}
Let us denote by $\FeasRt \subseteq \R^{\mydim}$ the set of feasible
solutions to the above problem~\eqref{model:node_t} with $\mydim
\define \sum_{i \in N} (k_i + l_i)+1+ 2n $.

We discuss in the following the required cuts that enhance the
approximation of~$\VBR(x)$ and $\pi(x)$ by $\proxy$ and $\piproxy$,
respectively. We start with the following definition.

\begin{definition}
  \label{def:non-ne-cut}
  For any node~$t$ of the search tree, let $(x^*,\lambda^*,
  \proxy^*,\piproxy^*)\in \FeasStrats\times \R\times
  \R^N\times \R^N$ be an integer-feasible
  node solution, i.e., an integer-feasible solution to
  \eqref{model:node_t}.
  Consider an arbitrary best-response vector $y^*$ \wrt $x^*$.
  Then, we call an inequality
  $c(x,\lambda,\proxy,\piproxy;
  x^*,\lambda^*,\proxy^*,\piproxy^*,y^*,\mappro,\addappro) \leq 0$,
  which is parameterized by $(x^*,\lambda^*, \proxy^*,\piproxy^*,
  y^*,\mappro,\addappro)$, an
  approximate-Nash-equilibrium cut (\cut[)] for node~$t$ if the
  following two properties are satisfied:
  \begin{enumerate}[label=(\roman*)]
  \item\label{def:non-ne-cut:1} It is satisfied by all points
    $(x, 0,\VBR(x),\pi(x)) \in \Nespiapp \cap B_t \cap C_t$.
  \item\label{def:non-ne-cut:2} It is violated by
    $(x^*,\lambda^*,\proxy^*,\piproxy^*)$.
  \end{enumerate}
  Here, we abbreviate the set of \eqtuple[s] via
  \begin{align*}
    \Nespiapp\define \Defset{(x,0, \VBR(x),\pi(x))
    \in \FeasStrats \times \R \times \R^N\times \R^N}{x \in\Nesapp}.
  \end{align*}
\end{definition}

In addition, \acut is said to be globally valid if it is
satisfied by all points $(x,0,\VBR(x),\pi(x)) \in \Nespiapp$.
Such a cut is then valid for any node~$t$ of the B\&C search tree.

The B\&C method now works as follows.
Starting at the root node, we solve the current node problem
\eqref{model:node_t}. In case that the problem is infeasible or the
optimal objective is larger than $0$,
there does not exist an $(\mappro,\addappro)$-NE in this node and we
prune it.
Otherwise, we check if the optimal node solution is integer-feasible and if not,
create new nodes as usual by branching on fractional integer
variables. Once we obtain an integer-feasible node solution
$(x^*,\lambda^*,\proxy^*,\piproxy^*)\in \FeasStrats\times \R\times
\R^N\times \R^N$, we check if~$x^*$ is actually an
$(\mappro,\addappro)$-NE. If so, we stop and return the
$(\mappro,\addappro)$-NE.
Otherwise, we use \acut to cut off the integer-feasible
point $(x^*,\lambda^*,\proxy^*,\piproxy^*)$ without removing any point
$(x,0,\proxy,\piproxy)\in \Nespiapp$ with $x$ being an
$(\mappro,\addappro)$-NE.
The procedure to process a node $t$ is described formally in
\Cref{algorithm:BC}.

\begin{algorithm}
  \caption{Processing Node~$t$}
  \label{algorithm:BC}
  \begin{algorithmic}[1]
    \State Solve \eqref{model:node_t}. \label{line:step_1}
    \If{\eqref{model:node_t} is infeasible \textbf{or} the optimal
      objective is strictly positive} \label{line: Infeasible}
    \State Prune the node.
    \Else
    \State Let $(x^*,\lambda^*,\proxy^*,\piproxy^*)$ be a solution to
    \eqref{model:node_t}. \label{line: NodeSol}
    \If{$x^* \notin \FeasStrats$} \label{line: xNotInX}
    \State Create two child nodes by branching on a fractional variable.
    \Else
    \State Determine $\VBR(x^*)$ and  obtain a
    solution~$y^*$. \label{line: SolveBR}
    \If{$\pi_i(x^*) \leq \mappro_i \VBR_i(x^*_{-i}) + \addappro_i $ for all $i\in
      N$} \label{line: NEFound}
    \Comment{$x^*$ is a $(\mappro,\addappro)$-NE}
    \State Return $x^*$ and stop the overall B\&C method.\label{line:
      return-x-star}
    \Else \Comment{$\proxy^* \nleq \VBR(x^*)$ or $\piproxy^* \ngeq \pi(x^*)$}
    \label{line: proxy>BR}
    \State Augment $C_t$ with an \cut \wrt $(x^*,\lambda^*, \proxy^*,\piproxy^*,
    y^*,\mappro,\addappro)$.
    \label{line: Cut}
    \State Go to Step~\ref{line:step_1}.
    \EndIf
    \EndIf
    \EndIf
  \end{algorithmic}
\end{algorithm}%
%
%

Note that as long as the introduced cuts result in closed sets~$C_t$,
the solution in Line~\ref{line: NodeSol} always exists if \eqref{model:node_t} is feasible.
Indeed, the feasible set~$\FeasRt$ of \eqref{model:node_t} is closed
since $C_t$, $B_t$, and $\FeasStratsRel$ are closed and all remaining
constraints are affine inequalities.
Although $\FeasRt$ is unbounded in general, its intersection with
every level set $\defset{(x,\lambda,\proxy,\piproxy)}{\lambda \leq c}$,
$c \in \R$, is bounded: $\FeasStratsRel$ is bounded and every point
of~$\FeasRt$ with $\lambda \leq c$ satisfies
\begin{equation*}
  \piproxy_i^- \leq \piproxy_i \leq \mappro_i(c + \proxy_i^+) + \addappro_i,
  \quad
  \frac{\piproxy_i^- - \addappro_i}{\mappro_i} - c \leq \proxy_i \leq \proxy_i^+,
  \quad
  \frac{\piproxy_i^- - \addappro_i}{\mappro_i} - \proxy_i^+ \leq \lambda \leq c
\end{equation*}
for all $i \in N$.
Hence, \eqref{model:node_t} attains its minimum.
Note that cuts involving the cost functions, such as
$\piproxy \geq \pi(x)$ in \Cref{sec:cuts-ne}, result in closed sets
due to the lower semi-continuity of~$\pi_i$ on~$\FeasStratsRel$.

In the remainder of this section, we prove the correctness of the B\&C
method, i.e., we show that if the method terminates, it yields an
$(\mappro,\addappro)$-NE $x^* \in \Nesapp$ or a certificate for the
non-existence of $(\mappro,\addappro)$-NEs.
It is clear that if a strategy profile~$x^*$ is returned by the
algorithm, then this strategy profile is an $(\mappro,\addappro)$-NE
since the condition in Line~\ref{line: NEFound} is fulfilled in this
case. Hence, the correctness follows from the next
\namecref{thm:correctness_algBC}.
\begin{theorem}
  \label{thm:correctness_algBC}
  If the B\&C algorithm terminates without finding an $(\mappro,\addappro)$-NE,
  then there does not exist an $(\mappro,\addappro)$-NE.
\end{theorem}
\begin{proof}
  We first make the following observation.
  In the root node, the feasible set contains the set $\Nespiapp$ as
  for every $(x,0,\VBR(x),\pi(x)) \in \Nespiapp$, we have
  $(x,0,\VBR(x),\pi(x)) \in \FeasStratsRel \times \R \times
  \R^N\times \R^N$, $\VBR(x) \leq \proxy^+$ as well as
    $\pi(x) \geq \piproxy^-$ by the validity of these bounds, and
  \begin{align*}
    \pi_i(x) \leq \mappro_i
    \VBR_i(x_{-i}) + \addappro_i
    \implies 0 \geq
    \frac{\pi_i(x)}{\mappro_i}- \VBR_i(x_{-i}) -\frac{\addappro_i}{\mappro_i}.
  \end{align*}
  Moreover, branching does not remove any point
    $(x,0,\VBR(x),\pi(x))$ with $x \in \FeasStrats$ since such points
    satisfy all integrality constraints.
  Hence, due to Condition~\ref{def:non-ne-cut:1} in
  \Cref{def:non-ne-cut}, the following invariant is true throughout
  the execution of the B\&C algorithm:
  The set $\Nespiapp$ is contained in the union of the feasible
  sets of the problems~\eqref{model:node_t} over all leaf nodes $t$ in
  the B\&C tree, i.e.,
  \begin{equation*}
    \Nespiapp
    \subseteq \bigcup_{t \text{ is a leaf}} \FeasRt.
  \end{equation*}
  Note that  pruned nodes are leafs of the B\&C tree as
  well.

  We argue in the following that $\FeasRt\cap \Nespiapp = \emptyset$
  holds for all leafs $t$ in the case of the B\&C algorithm
  terminating without finding an $(\mappro,\addappro)$-NE.
  It then follows directly by the above invariant that there does not
  exist any $(\mappro,\addappro)$-NE.
  If the B\&C algorithm terminates without finding an $(\mappro,\addappro)$-NE,
  every node~$t$ was ultimately pruned, i.e., the condition in
  Line~\ref{line: Infeasible} was met and Problem~\eqref{model:node_t}
  became either infeasible or had a strictly positive optimal
  objective value.
  In the former case, it is clear that $\FeasRt\cap \Nespiapp =
  \emptyset$ holds because of $\FeasRt = \emptyset$. Hence, consider
  a pruned leaf node $t$ with corresponding optimal objective value
  being strictly positive. Then, any $(x,\lambda, \VBR(x),\pi(x))\in \FeasRt$
  has $\lambda >0$ and, subsequently, $(x,\lambda, \VBR(x),\pi(x)) \notin
  \Nespiapp$ holds.
  Thus, the proof is finished.
\end{proof}

So far, we have shown the correctness of our B\&C method for arbitrary
\cut[s].
In the subsequent section, we derive under suitable conditions
the existence of such \cut[s]
and also give sufficient conditions under which they lead to finite
termination of the B\&C method in the NEP setting and, in
combination with a fallback strategy, in the GNEP setting.


\section{Cuts and Finite Termination}
\label{sec:cuts-finite-termination}

We now investigate the existence of \cut[s]
and finite termination of our B\&C method. 
To this end, we require the following useful observation
regarding \Cref{algorithm:BC}:
Whenever a cut needs to be added, at least one of the proxy variables
must exhibit slack with respect to its corresponding approximation quantity.
This insight is crucial in order to show the existence of a proper cut, tightening
the respective proxy variable.

\begin{lemma}\label{lem: proxy>BR}
	Suppose that the algorithm enters the else-part in Line~\ref{line:
		proxy>BR}.
	Then, at least one of the following two statements is true:
	\begin{enumerate}[label=(\roman*)]
		\item There exists an $i \in N$ with $\proxy^*_i >
		\VBR_i(x_{-i}^*)$.\label{enum: lem: proxy>BR:1}
		\item There exists an $i \in N$ with $\piproxy^*_i <
		\pi_i(x^*)$. \label{enum: lem: proxy>BR:2}
	\end{enumerate}
\end{lemma}
\begin{proof}
	We argue in the following that \ref{enum: lem: proxy>BR:1} has to be
	fulfilled, if \ref{enum: lem: proxy>BR:2} is not satisfied.
	
	Since the condition in Line~\ref{line: Infeasible} was not met,
	the objective value of $(x^*,\lambda^*,\proxy^*,\piproxy^*)$ is non-positive,
	i.e., $\lambda^*\leq 0$.
	In particular, by the feasibility of this point, we get  for all $j
	\in N$ that
	\begin{align*}
		0\geq \lambda^*\geq \frac{\piproxy_j^*}{\mappro_j}- \proxy_j^* -
		\frac{\addappro_j}{\mappro_j}
		\implies \proxy_j^* \geq \frac{\piproxy_j^*}{\mappro_j}-
		\frac{\addappro_j}{\mappro_j}.
	\end{align*}
	Using that \ref{enum: lem: proxy>BR:2} is not fulfilled in the above
	yields
	\begin{align}\label{eq: proxy>VBR}
		\proxy_j^* \geq \frac{\piproxy_j^*}{\mappro_j}-
		\frac{\addappro_j}{\mappro_j} \geq \frac{\pi_j(x^*)}{\mappro_j}-
		\frac{\addappro_j}{\mappro_j}.
	\end{align}
	
	Since the condition in Line~\ref{line: NEFound} was not satisfied,
	there exists an $i \in N$ with $\pi_i(x^*)> \mappro_i
	\VBR_i(x_{-i}^*) + \addappro_i$. Using the above inequality
	\eqref{eq: proxy>VBR} for $j=i$ then yields
	\begin{align*}
		\proxy_i^* \geq  \frac{\pi_i(x^*)}{\mappro_i} -
		\frac{\addappro_i}{\mappro_i}  > \frac{\mappro_i\VBR_i(x_{-i}^*) +
			\addappro_i}{\mappro_i}- \frac{\addappro_i}{\mappro_i}
		= \VBR_i(x_{-i}^*)
	\end{align*}
	and the desired statement follows.
\end{proof}

The remaining section is split into two subsections. 
First, we consider the general case of GNEPs and prove the
existence of \cut[s] via intersection cuts or, for the proxy~$\piproxy$ and convex cost functions, via gradients under suitable assumptions. Afterwards, we come to the
special case of standard NEPs. There, we prove that 
\cut[s] always exist and determine conditions under which our B\&C method terminates in finite time. 
As discussed in \Cref{rem: Fallback}, the latter cuts can also be
used within a fallback strategy to ensure finite termination in the
GNEP case.

\subsection{Generalized Nash Equilibrium Problems}
\label{sec:cuts-gne}

In this section, we derive \cut[s] via ICs or, for the
  proxy~$\piproxy$ in the case of convex cost functions, directly via
  gradients.
With the framework presented here, we are able to guarantee the
existence of these \cut[s] under the following  set of assumptions.
\begin{assumption}\label{ass: IC}
  For every player $i\in N$, the following holds true.
  \begin{enumerate}[label=(\roman*)]
  \item The constraint function $g_i$ is affine. \label{ass: IC1}
  \item\label{ass: IC2} One of the following statements holds:
    \begin{enumerate}[label=(\alph{enumii}),ref = (\roman{enumi}\alph*)]
    \item \label{ass: IC2: conv} Her cost function is differentiable
      and convex in the entire strategy profile, i.e.,
      \begin{align*}
        \pi_i : \Rall \to \R, \ x\mapsto \pi_i(x),
      \end{align*}
      is differentiable and convex.
    \item \label{ass: IC2: conc} Her cost function is concave in the
      entire strategy profile and affine in the rivals' strategies,
      i.e.,
      \begin{align*}
        \pi_i : \Rall \to \R, \ x\mapsto \pi_i(x), \text{ is concave,}
      \end{align*}
      and for all $x_{i} \in \R^{l_i+k_i}$,
      \begin{align*}
        \pi_i(x_i,\cdot):\Rallmini \to \R,  \ x_{-i}\mapsto
        \pi_i(x_i,x_{-i}) \text{ is affine.}
      \end{align*}
    \item \label{ass: IC2: Lip} Her cost function is convex in the rivals' strategies
      and Lipschitz continuous \wrt the
      constant $\Lip_i >0$ in the entire strategy profile, i.e.,
      \begin{align*}
        \abs{\pi_i(x)-\pi_i(x')}  \leq \Lip_i \norm{x-x'} \quad  \text{ for all } x,x' \in \Rall,
      \end{align*}
      where $\norm{\cdot}$ denotes the euclidean norm.
    \end{enumerate}
  \end{enumerate}
\end{assumption}

\begin{remark}
  \begin{enumerate}[label=(\roman*)]
  \item For the ease of presentation, we assume the Lipschitz continuity in
    \Cref{ass: IC}\ref{ass: IC2: Lip} to hold globally. Yet,
    one can easily verify that it is sufficient that the Lipschitz continuity
    holds on an arbitrary open neighborhood of $\FeasStratsRel$.
    With this relaxation, \Cref{ass: IC}\ref{ass: IC2: conv}
      and~\ref{ass: IC2: conc} are special cases of
      \Cref{ass: IC}\ref{ass: IC2: Lip}: in both cases, $\pi_i$ is
      convex in the rivals' strategies and, being convex or concave
      on~$\Rall$, Lipschitz continuous on every bounded open
      neighborhood of the compact set~$\FeasStratsRel$; see
      \Cref{ass:GeneralAss}\ref{ass:GeneralAss:Poly}.
      However, the \Nefreeset $\freeset[\piproxy,\Lip]_i(x^*)$ that we
      use below under \Cref{ass: IC}\ref{ass: IC2: Lip} replaces
      $\pi_i(x)$ by its lower bound $\pi_i(x^*) - \Lip_i\norm{x-x^*}$,
      whereas $\freeset[\piproxy,\conc]_i$ and the gradient cut from
      \Cref{lem:gradient-cut} use $\pi_i(x)$ itself and its
      linearization at~$x^*$, respectively.
      Since Lipschitz continuity is a rather weak property, the IC
      derived from $\freeset[\piproxy,\Lip]_i(x^*)$ is typically rather
      weak as well.
      Hence, it makes sense to use the cuts based on
      \Cref{ass: IC}\ref{ass: IC2: conv} or~\ref{ass: IC2: conc}
      whenever possible.
  \item   \Cref{ass: IC}\ref{ass: IC2: conc} implies that the cost function of
      player~$i$ is separable in the sense that
      \begin{equation*}
        \pi_i(x) = a_i(x_i) + c_i^\top x_{-i}
        \quad \text{for all } x \in \Rall
      \end{equation*}
      holds for a concave function~$a_i\colon \R^{k_i+l_i} \to \R$ and a
      vector~$c_i \in \Rallmini$.
      Hence, the rivals' strategies influence the costs of player~$i$ only
      additively and her best responses depend on~$x_{-i}$ only via her
      constraints.
  \end{enumerate}
\end{remark}

\jscom[journal]{We could actually also relax the affineness to
  convexity of $x_{-i}\mapsto g_i(y_i,x_{-i})$ for all $y_i$! We would
  then need to relax also the constraints in the node problems via a
  box constraint for example.}

For the remainder of this section, consider the situation of
\Cref{def:non-ne-cut} and fix the corresponding integer-feasible
solution $(x^*,\lambda^*,\proxy^*,\piproxy^*)$ and a corresponding
best response~$y^*$.
With this at hand, we derive sufficient conditions to define
\acut via an IC or a gradient evaluation.
To this end, we first observe that under \Cref{ass: IC},
the root-node problem is a linear program. In particular,
since ICs and the gradient cuts from \Cref{lem:gradient-cut} are affine, any node problem during the B\&C
remains a linear program if we only use these cuts.
In this situation, the existence of an IC boils down to the existence of an \emph{\Nefreeset}
$\freeset(x^*,\lambda^*,\proxy^*,\piproxy^*)$, i.e., a convex set
whose interior contains $(x^*,\lambda^*,\proxy^*,\piproxy^*)$
but no \eqtuple in $\Nespiapp \cap \FeasRt$.
Note that without the condition of the node problem being a linear program, optimal solutions need not be attained at a vertex of the underlying polyhedral feasible set and the construction of an IC requires the existence and construction of a cone pointed at the node solution that has to be cut off.
In our case, we select an optimal vertex in
  Line~\ref{line: NodeSol} of \Cref{algorithm:BC} and use its
LP basis cone.

We will introduce in the following two types of \Nefreeset,
strengthening the approximation of $\VBR(x)$ and
$\pi(x)$ by $\proxy$ and $\piproxy$, respectively.
For the \Nefreeset \wrt $\proxy$, we define for all
$i \in N$ the set
\begin{align*}
  \freeset[\proxy]_i(x^*,y^*)
  \define
  &\Defset{(x,\lambda,\proxy,\piproxy) \in\R^\mydim}{\proxy_i >
    \pi_i(y_i^*,x_{-i}), \, y_i^* \in X_i(x_{-i})}
  \\
  =
  &\Defset{(x,\lambda,\proxy,\piproxy) \in \R^\mydim}{ \proxy_i >
    \pi_i(y_i^*,x_{-i}), \, g_i(y_i^*,x_{-i})\leq 0}.
\end{align*}
For the \Nefreeset \wrt $\piproxy$, we define two different sub-types
of sets,
which we prove to be \Nefreeset[s] depending on whether
  \Cref{ass: IC}\ref{ass: IC2: conc} or \Cref{ass: IC}\ref{ass: IC2: Lip}
  holds; under \Cref{ass: IC}\ref{ass: IC2: conv}, we use gradient cuts
  instead; see \Cref{lem:gradient-cut}.
For all $i\in N$, let us define
\begin{align*}
  \freeset[\piproxy,\conc]_{i}
  & \define \Defset{(x,\lambda,\proxy,\piproxy) \in
    \R^\mydim}{\piproxy_i < \pi_i(x) } \text{ and } \\
  \freeset[\piproxy,\Lip]_{i}(x^*)
  & \define \Defset{(x,\lambda,\proxy,\piproxy) \in
    \R^\mydim}{\piproxy_i +L_i\norm{x-x^*} < \pi_i(x^*)}.
\end{align*}
For these sets, we get the following convexity statements:

\begin{lemma}\label{lem: FreeSetConv}
  The set $\freeset[\piproxy,\Lip]_{i}(x^*) $ is convex.
  If \Cref{ass: IC} holds, the set $\freeset[\proxy]_i(x^*,y^*)$ is
  also convex. Moreover, if
    \Cref{ass: IC}\ref{ass: IC2: conc} holds, the set
  $\freeset[\piproxy,\conc]_{i}$ is convex.
\end{lemma}
\begin{proof}
  Convexity of $\freeset[\piproxy,\Lip]_{i}(x^*) $ is a  trivial
  consequence of the triangle inequality.
  For the convexity of $\freeset[\proxy]_i(x^*,y^*)$, we argue as follows:
  \begin{description}
  \item[{$\freeset[\proxy]_i(x^*,y^*)$}]
    By rewriting the first condition of
    $\freeset[\proxy]_i(x^*,y^*)$ via $\pi_i(y_i^*,x_{-i})- \proxy_i <
    0$,  it follows that
    this is a convex restriction under \Cref{ass: IC}\ref{ass: IC2: conv} and \Cref{ass: IC}\ref{ass: IC2: Lip}
    or an affine restriction under \Cref{ass: IC}\ref{ass: IC2: conc}. Hence, since
    either one of these conditions has to hold under \Cref{ass: IC}, this
    restriction always leads to a convex one.
    Since the second condition is affine  under \Cref{ass:
      IC}\ref{ass: IC1}, the convexity of
    $\freeset[\proxy]_i(x^*,y^*)$ follows.
  \end{description}
  Next, we show that $\freeset[\piproxy,\conc]_{i}$ is convex if \Cref{ass:
    IC}\ref{ass: IC2: conc} is fulfilled.
  \begin{description}
  \item[{$\freeset[\piproxy,\conc]_{i}$}] This is an immediate
    consequence of the assumed concavity in \Cref{ass: IC}\ref{ass:
      IC2: conc}.\qedhere
  \end{description}
\end{proof}
Next, we show that these sets contain the optimal solution for the
subsets
\begin{align*}
  N^\proxy(x^*,\proxy^*)\define \Defset{i \in N}{\proxy^*_i >
  \VBR_i(x^*_{-i})} \quad \text{and} \quad
  N^\piproxy(x^*,\piproxy^*)\define\Defset{i \in N}{\piproxy^*_i <
  \pi_i(x^*)}
\end{align*}
of players.

\begin{lemma}\label{lem: FreeSet}
  Under \Cref{ass: IC}, the following statement holds:\hfill
  \begin{enumerate}[label=(\roman*)]
  \item $(x^*,\lambda^*,\proxy^*,\piproxy^*)\in
    \freeset[\proxy]_i(x^*,y^*)$ for any $i \in
    N^\proxy(x^*,\proxy^*)$. Moreover, $\freeset[\proxy]_i(x^*,y^*)$
    does not contain any point of the intersection $\Nespiapp\cap
    \FeasRt$ for all $i \in N$.
  \end{enumerate}
  Moreover,
  depending on whether
    \Cref{ass: IC}\ref{ass: IC2: conc} or
    \Cref{ass: IC}\ref{ass: IC2: Lip} holds for $i \in	N^\piproxy(x^*,\piproxy^*)$, one of the
  following statements holds:
  \begin{enumerate}[label=(\roman*),resume]
  \item $\freeset[\piproxy,\conc]_i$ contains
    $(x^*,\lambda^*,\proxy^*,\piproxy^*)$ in its
    interior, and does not contain any point of the
    intersection $\Nespiapp\cap \FeasRt$, and is hence an \Nefreeset[.]
  \item  $\freeset[\piproxy,\Lip]_i(x^*)$ contains
    $(x^*,\lambda^*,\proxy^*,\piproxy^*)$  in its interior, and does not contain any point of
    the intersection $\Nespiapp\cap \FeasRt$, and is hence an \Nefreeset[.]
  \end{enumerate}
\end{lemma}
\begin{proof}
  We start with the statements about $ \freeset[\proxy]_i(x^*,y^*)$.
  \begin{enumerate}[label=(\roman*)]
  \item It holds $(x^*,\lambda^*,\proxy^*,\piproxy^*) \in
    \freeset[\proxy]_i(x^*,y^*)$ for any $i \in
    N^\proxy(x^*,\proxy^*)$ because $y_i^* \in \argmin_{y_i \in
      X_i(x_{-i}^*)} \pi_i(y_i,x_{-i}^*)$ implies
    $\proxy^*_i>\VBR_i(x^*_{-i}) = \pi_i(y_i^*,x^*_{-i})$ and $y^*_i \in
    X_i(x^*_{-i})$.
    Moreover, for any $(\bar{x},0,\VBR(\bar{x}),\pi(\bar{x}))\in
    \Nespiapp\cap \FeasRt$ and $i \in N$ with $y^*_i \in
    X_i(\bar{x}_{-i})$, we have that
    \begin{align*}
      \VBR_i(\bar{x}_{-i}) =
      \min_{y_i \in X_i(\bar{x}_{-i})} \pi_i(y_i,\bar{x}_{-i}) \leq
      \pi_i(y_i^*,\bar{x}_{-i})
    \end{align*}
    holds, showing that $(\bar{x},0,\VBR(\bar{x}),\pi(\bar{x})) \notin
    \freeset[\proxy]_i(x^*,y^*)$.
  \end{enumerate}
  Next, we show that,  for $i \in	N^\piproxy(x^*,\piproxy^*)$, the statements about
  $\freeset[\piproxy,\conc]_i$ are valid if \Cref{ass: IC}\ref{ass:
    IC2: conc} is true
  and the statements about
  $\freeset[\piproxy,\Lip]_i(x^*)$ are valid if \Cref{ass: IC}\ref{ass:
    IC2: Lip} is true.
  \begin{enumerate}[label=(\roman*),resume]
  \item The statements are an immediate consequence of the definitions
    of the sets $\freeset[\piproxy,\conc]_i$ and $
    N^\piproxy(x^*,\piproxy^*)$ together with $\pi_i$ being
    continuous by the assumed concavity in \Cref{ass: IC}\ref{ass:
      IC2: conc}.
  \item 	The property that $(x^*,\lambda^*,\proxy^*,\piproxy^*)$ is
    contained in the interior of $\freeset[\piproxy,\Lip]_i(x^*)$ for
    any $i \in N^\piproxy(x^*,\piproxy^*)$  follows immediately by
    definition.
    Hence, consider an arbitrary
    $(\bar{x},0,\VBR(\bar{x}),\pi(\bar{x}))\in \Nespiapp\cap
    \FeasRt$. Then, we have by the assumed Lipschitz continuity  in \Cref{ass: IC}\ref{ass:
      IC2: Lip}
    that
    \begin{align*}
      \pi_i(\bar{x}) \geq \pi_i(x^*) -\Lip_i \norm{\bar{x}-x^*} \Leftrightarrow 	\pi_i(\bar{x}) +  \Lip_i \norm{\bar{x}-x^*}\geq \pi_i(x^*)
    \end{align*}
    demonstrating that
    $(\bar{x},0,\VBR(\bar{x}),\pi(\bar{x}))\notin\freeset[\piproxy,\Lip]_i(x^*)$
    as the defining constraint of the latter set is violated. \qedhere
  \end{enumerate}
\end{proof}
  Under \Cref{ass: IC}\ref{ass: IC2: conv}, we do not need an IC for
  $i \in N^\piproxy(x^*,\piproxy^*)$ but can directly use an affine
  inequality as \acut[.]
  Here and in what follows, we denote by $\nabla\pi_i(x)$ the gradient
  of~$\pi_i$ at~$x$.
  \begin{lemma}\label{lem:gradient-cut}
    Let $i \in N^\piproxy(x^*,\piproxy^*)$ and suppose that the cost
    function~$\pi_i$ satisfies \Cref{ass: IC}\ref{ass: IC2: conv}.
    Then, the inequality
    \begin{equation*}
      \piproxy_i \geq \pi_i(x^*) + \nabla\pi_i(x^*)^\top(x-x^*)
    \end{equation*}
    yields \acut[,] which we call a \emph{gradient cut}.
    Moreover, this cut is globally valid and independent of the
      approximate parameter values $(\alpha,\beta)$.
  \end{lemma}
  \begin{proof}
    The gradient cut is violated by $(x^*,\lambda^*,\proxy^*,\piproxy^*)$
    since $\piproxy^*_i < \pi_i(x^*)$ holds by the definition
    of~$N^\piproxy(x^*,\piproxy^*)$.
    Moreover, the convexity of~$\pi_i$ yields
    \begin{equation*}
      \pi_i(\bar{x}) \geq \pi_i(x^*) + \nabla\pi_i(x^*)^\top(\bar{x}-x^*)
      \quad \text{for all } \bar{x} \in \Rall.
    \end{equation*}
    Hence, every \eqtuple $(\bar{x},0,\VBR(\bar{x}),\pi(\bar{x})) \in
    \Nespiapp$ for any $(\alpha,\beta)$ satisfies the gradient cut, i.e., the gradient cut is
    \acut and, in particular, globally valid and independent of
      the approximate parameter values $(\alpha,\beta)$.
  \end{proof}

The set
$\freeset[\proxy]_i(x^*,y^*)$ is, in general, not suitable for
deriving ICs as it is not guaranteed that
$(x^*,\lambda^*,\proxy^*,\piproxy^*)$ belongs to its interior.
This motivates us to define the following
extended version:
\begin{align*}
  \freeset[\proxy,\varepsilon]_i(x^*,y^*) \define
  \Defset{(x,\lambda,\proxy,\piproxy) \in \R^\mydim}{\proxy_i \geq
  \pi_i(y_i^*,x_{-i}), \, g_i(y_i^*,x_{-i})\leq \varepsilon
  \mathbf{1}},
\end{align*}
where $\varepsilon>0$ and   $\mathbf{1}$ denotes the vector of all ones (in appropriate
dimension).
Provided that no point in $\Nespiapp\cap\FeasRt$ is contained in
the interior of this extended set, it follows from \Cref{lem:
  FreeSetConv,lem: FreeSet} that  $\freeset[\proxy,\varepsilon]_i(x^*,y^*)$
is an \Nefreeset under \Cref{ass: IC}.
In this regard, \textcite{DSHS25} provided a sufficient
  condition, which carries directly over to our setting and is stated
  in the following lemma.
  In \Cref{rem: SuffConNEfree}, which can also be found in
  \textcite{DSHS25}, we then discuss when this condition
  is satisfied.

\begin{lemma}\label{lem: SuffConNEfree}
  Let \Cref{ass: IC} hold, let $i \in N^\proxy(x^*,\proxy^*)$,
    and assume that, for every $\bar{x} \in \Nesapp$ and every
    $j \in [m_i]$, we have $g_{i,j}(y_i^*,\bar{x}_{-i}) \in \Z$ or
    $g_{i,j}(y_i^*,\bar{x}_{-i}) \leq 0$.
    Then, $\freeset[\proxy,\varepsilon]_i(x^*,y^*)$ with $\varepsilon = 1$
    does not contain any point of $\Nespiapp\cap C_t\cap B_t$ in its
    interior.
\end{lemma}
  \begin{remark}\label{rem: SuffConNEfree}
    \begin{enumerate}[label=(\roman*)]
    \item The condition $g_{i,j}(y_i^*,\bar{x}_{-i}) \leq 0$ in
      \Cref{lem: SuffConNEfree} is trivially satisfied for any
      constraint that does not depend on the rivals' strategies since
      $y_i^*$ is a best response \wrt~$x^*$ and, hence, in particular
      feasible.
    \item Analogous assumptions are made in the respective literature
      on mixed-integer bilevel optimization; see, e.g.,
      \textcite{fischetti2018use}, \textcite{Lozano_Smith:2017}, or
      \textcite{Horländer_et_al:2024}.
      In practice, very often the condition
      $g_{i,j}(y_i^*,\bar{x}_{-i}) \in \Z$ for all
      $\bar{x} \in \Nesapp$ in \Cref{lem: SuffConNEfree} is satisfied
      due to the constraints being affine, all problem data being
      integer (or, w.l.o.g., rational), and all players' variables being
      integer as well, which is easy to check in practice.
      In particular, any one of the following conditions (for a suitable
      integral matrix~$A_i$ and vector~$b_i$) implies
      $g_{i,j}(y_i^*,\bar{x}_{-i}) \in \Z$ for all
      $\bar{x} \in \Nesapp$ and all $j \in [m_i]$, where we denote by
      $x_{i}^\inte \define (x_{i,1}, \dotsc, x_{i,k_i})$ the integer
      components of player~$i$'s strategy and analogously by
      $x_{-i}^\inte$ the rivals' integer components:
      \begin{enumerate}[label=(\alph*)]
      \item $y_i^*$ is integral and $g_i(y_i^*,\bar{x}_{-i}) = A_i
        (y_i^*,\bar{x}_{-i}^\inte) - b_i$ for all
        $\bar{x} \in \Nesapp$;
      \item $g_i(y_i^*,\bar{x}_{-i}) = A_i
        ((y_i^*)^\inte,\bar{x}_{-i}^\inte) - b_i$ for all
        $\bar{x} \in \Nesapp$;
      \item $y_i^*$ and all $\bar{x} \in \Nesapp$ are integral and
        $g_i(y_i^*,\bar{x}_{-i}) = A_i (y_i^*,\bar{x}_{-i}) - b_i$
        holds for all $\bar{x} \in \Nesapp$.
      \end{enumerate}
    \end{enumerate}
  \end{remark}

Concluding, by \Cref{lem: proxy>BR} and the above lemmata, we can guarantee
under suitable assumptions the existence of \cut[s] via an IC or a gradient cut
whenever a cut is needed.
\begin{theorem}\label{thm: GNEPs}
  In the situation of Line~\ref{line: Cut} of \Cref{algorithm:BC},
  suppose that the optimal node solution is chosen to be a
    vertex of $\FeasRt$. Then,
  we can guarantee the existence of an ANE-cut
  via intersection cuts or gradient cuts under
  \Cref{ass: IC} and the fulfillment of
  the conditions in \Cref{lem: SuffConNEfree}.
\end{theorem}
For an explicit description on how to construct an IC, we refer to
standard literature on IC theory such as
\textcite[Chapter~6]{Conforti-et-al:2014} or the description presented by \textcite{DSHS25}.

Let us close this section with a remark on the finite termination of
our B\&C method for GNEPs.

\begin{remark}\label{rem: Fallback}
  ICs may become arbitrarily shallow so that the B\&C method relying on
  the cuts from \Cref{thm: GNEPs} need not terminate in finite time.
  For exact equilibria, \textcite[Section~4.3]{DSHS25} resolve this
  issue via a fallback strategy, which carries over to our setting:
  Whenever the same integer profile has been part of the node solution
  in more than a prescribed number of iterations of
  \Cref{algorithm:BC}, we branch such that one child node fixes the
  integer components of all players to this profile, while the
  remaining child nodes cover all other integer profiles.
  In every node~$t$ with fixed integer components, we add the
  constraints $\piproxy \geq \pi(x)$, which ensure
  $N^\proxy(x^*,\proxy^*) \neq \emptyset$ by \Cref{lem: proxy>BR}
  whenever a cut is needed, and only use the best-response
  cuts~\eqref{eq: PlayerCut} introduced in \Cref{sec:cuts-ne} below.
  Suppose now that \textcite[Assumption~4.16]{DSHS25} is satisfied,
  i.e., for every fixed integer profile, (i) the strategy
  set~$X_i(x_{-i})$ of every player~$i \in N$ does not depend on the
  rivals' continuous strategies and (ii) every player has only finitely
  many distinct best-response sets; see
  \textcite[Remark~4.17]{DSHS25} for sufficient conditions.
  Then, the cut~\eqref{eq: PlayerCut} is \acut for node~$t$ for every
  $i \in N^\proxy(x^*,\proxy^*)$.
  In particular, both conditions are trivially satisfied if all variables are integer.
  By (i) and (ii), an analogous argument as in the proof of
  \Cref{thm: FinitelyManyCutsGeneral} then shows that node~$t$ is
  processed after finitely many iterations.
  Since there are only finitely many integer profiles, it follows as in
  \textcite[Theorem~4.19]{DSHS25} that the B\&C method with this
  fallback strategy terminates in finite time, provided that \acut is
  available whenever Line~\ref{line: Cut} is reached in one of the
  remaining nodes, e.g., under the conditions of \Cref{thm: GNEPs}.
  Note that the node problems with fixed integer components are, in
  general, no longer linear programs.
  Yet, no ICs are required in these nodes and their feasible sets
  remain closed since the cost functions are continuous under
  \Cref{ass: IC}.
\end{remark}


\subsection{Standard Nash Equilibrium Problems}
\label{sec:cuts-ne}

We now come to the special case of $G$ being a standard NEP,
i.e., $X_i(x_{-i}) \equiv X_i$ for some fixed strategy set~$X_i$
given by $X_i \define \defset{x_i \in \Z^{k_i}\times \R^{l_i}}{g_i(x_i)
  \leq 0}$.
Note that the set of feasible strategy profiles is then
given by $\FeasStrats= \prod_{i \in N}X_i$.

We derive in the following a cut specifically tailored
to the NEP setting and suitable conditions under which the
B\&C method finitely terminates. Since we assume that the
relaxation~$\FeasStratsRel$ is bounded, there are only finitely many
different possible combinations of feasible integral strategy
components.
In particular, there may only appear finitely many nodes in the
B\&C search tree.
Hence, it is sufficient to show that \Cref{algorithm:BC} processes
every node in finite time to show that the overall B\&C method
finitely terminates.

Let us now come to the promised cuts.
A fundamental difference in comparison to the
general GNEP case and the ICs used there is that we do
not require optimal solutions to the node problem
to be a vertex of a polyhedral set to derive
a suitable cut. In this regard, the cuts themselves do not
need to be affine either. In particular, the
approximation of the cost function via $\piproxy$
becomes unnecessary and the addition of these variables becomes
obsolete. In order to keep the same notation, we
will still consider the problem \eqref{model:node_t}
for a node problem but assume that we already have employed
the cuts $\piproxy \geq \pi(x)$ in the root node.
Note that in this case, \Cref{lem: proxy>BR}
shows that $N^\proxy(x^*,\proxy^*)\neq \emptyset$ holds
in the situation considered there.
Moreover, we assume throughout this section that, for every $i \in N$
and every $y_i \in X_i$, the function $x_{-i} \mapsto
\pi_i(y_i,x_{-i})$ is upper semi-continuous on the projection
of~$\FeasStratsRel$ onto the strategy space of the rivals of
player~$i$, which holds, e.g., if $\pi_i$ is continuous
on~$\FeasStratsRel$.
Then, the cuts introduced below result in closed sets~$C_t$ so that the
node problems admit an optimal solution whenever they are feasible; see
the discussion after \Cref{algorithm:BC}.

\begin{lemma}\label{lem: Cut}
  In the situation of \Cref{def:non-ne-cut}, the best-response-cut given
  by
  \begin{align}
    \label{eq: PlayerCut}
    c(x,\lambda,\proxy,\piproxy; x^*,\lambda^*, \proxy^*,\piproxy^*,
    y^*,\mappro,\addappro) \define
    c_i(x,\lambda,\proxy,\piproxy;y^*) \define \proxy_i -
    \pi_i(y_i^*,{x}_{-i})
    \leq 0
  \end{align}
  yields \acut for every $i \in N^\proxy(x^*,\proxy^*)$.
  It also holds $N^\proxy(x^*,\proxy^*) \neq \emptyset$ in the situation of
  Line~\ref{line: proxy>BR} by \Cref{lem: proxy>BR}, i.e., we can
  always use \acut of the form~\eqref{eq: PlayerCut}.
  Moreover, these cuts are satisfied by all
  $(x,0,\VBR(x),\piproxy)\in \FeasStrats\times \R\times \R^N \times \R^N$,
  i.e. , they are globally valid, and independent of the approximate
  parameter values $(\alpha,\beta)$.
\end{lemma}
\begin{proof}
  Consider an arbitrary tuple $(\bar{x},0,\VBR(\bar{x}),\pi(\bar{x})) \in
  \FeasStrats\times \R\times \R^N \times \R^N$ and $i \in N$.
  Then, we have
  \begin{equation*}
    \VBR_i(\bar{x}_{-i})=
    \min_{y_i \in X_i} \pi_i(y_i,\bar{x}_{-i}) \leq
    \pi_i(y_i^*,\bar{x}_{-i}).
  \end{equation*}
  Hence, $c(\bar{x},0,\VBR(\bar{x}),\piproxy;y^*)\leq 0$ and the cut
  fulfills Condition~\ref{def:non-ne-cut:1} of \Cref{def:non-ne-cut} independently of $B_t$, $C_t$, and the approximate parameter values $(\alpha,\beta)$.

  Condition~\ref{def:non-ne-cut:2} is an immediate
  consequence of the definition of $N^\proxy(x^*,\proxy^*)$ and
  $\VBR_i(x_{-i}^*) = \pi_i(y_i^*,x_{-i}^*)$.
\end{proof}

In the following, we derive sufficient conditions under which
\Cref{algorithm:BC} (and, hence, the overall B\&C method)
terminates in finite time when using the cuts introduced in \Cref{lem:
  Cut}.
The following Theorem~\ref{thm: FinitelyManyCutsGeneral} provides an
abstract sufficient condition, which was shown in
\textcite[Lemmas~4.7 and~4.8]{DSHS25} to be satisfied
for the important two special cases in which
\begin{enumerate}[label=(\roman*)]
\item the players' cost functions are concave in their own continuous
  strategies and their strategy space is polyhedral or
\item the players' cost function only depend on their own strategy and
  the rivals integer strategy components.
\end{enumerate}

To state the promised theorem, we introduce the following terminology.
Let us denote by
\begin{align*}
  \BR_i(x_{-i}) \define \argmin \Defset{\pi_i(y_i,x_{-i})}{y_i\in X_i}
\end{align*}
the set of best responses to $x_{-i} \in \FeasStrats[-i]$.
Moreover, let us define the set of all possible best response sets
by
\begin{align*}
  \BRcom_i \define \Defset{\BR_i(x_{-i})}{x_{-i} \in
  \FeasStrats[-i]} \subseteq \powset(\FeasStrats[i]),
\end{align*}
where we denote
by $\powset(M)$ the power set of a set $M$.

\begin{theorem}\label{thm: FinitelyManyCutsGeneral}
  Suppose that $\abs{\BRcom_i}$, $i\in N$, are finite.
  If we use the \cut \eqref{eq: PlayerCut} from \Cref{lem: Cut}
  in Line~\ref{line: Cut} of \Cref{algorithm:BC}, then
  \Cref{algorithm:BC} terminates after a finite number of steps.
\end{theorem}

The following lemma states a general statement about
how the cuts in~\eqref{eq: PlayerCut}
act on the feasible sets. From this, the claim of
\Cref{thm: FinitelyManyCutsGeneral} follows almost immediately.

\begin{lemma}\label{lem: FinitelyManyCutsGeneral}
  Consider $i \in N$, $j \in \N$, a sequence of points
  $(x^*_s,\proxy^*_s),s \in [j]\define\{ 1,\ldots,j\}$, with
  corresponding best responses $y^*_{s,i} \in
  \BR_i((x_s^*)_{-i}),s\in[j]$, and corresponding sets
  \begin{align}
    \label{eq: FinitelyManyCutsGeneral}
    C^s_i\define \Defset{(x,\proxy)\in \FeasStrats \times \R^N}{\proxy_i
    \leq \pi_i(y^*_{s',i},x_{-i}), \, s'\in[s-1]}, \ s\in[j].
  \end{align}
  Then, if $(x^*_j,\proxy^*_j)\in C^j_i$ holds and if there exists $\hat{s} \in
  [j-1]$ with  $y^*_{\hat{s},i} \in \BR_i((x^*_j)_{-i})$, we have $i
  \notin N^\proxy(x^*_j,\proxy^*_j)$.
  In particular, if for every $s \in [j]$,  $(x^*_s,\proxy^*_s) \in
  C^s_i$ and $i \in N^\proxy(x^*_s,\proxy^*_s)$ holds, we obtain $j \leq
  \abs{\BRcom_i}$.
\end{lemma}
\begin{proof}
  Assume that  $(x^*_j,\proxy^*_j)\in C^j_i$ and there exists $\hat{s}
  \in [j-1]$ with
  $y^*_{\hat{s},i} \in \BR_i((x^*_j)_{-i})$.
  Then, we have
  \begin{align*}
    (\proxy_j^*)_i
    \leq \pi_i(y^*_{{\hat{s}},i},(x^*_j)_{-i})
    =  \VBR_i((x^*_j)_{-i}).
  \end{align*}
  The inequality follows from $(x^*_j,\proxy^*_j)\in C^j_i$ and the
  equality holds due to $y^*_{{\hat{s}},i} \in \BR_i((x^*_j)_{-i})$.

  Now assume that for every $s \in [j]$, $(x^*_s,\proxy^*_s) \in C^s_i$
  and  $i \in N^\proxy(x^*_s,\proxy^*_s)$ holds.
  Consider two arbitrary sequence indices $s_1<s_2\leq j$. By applying
  the first part of the \namecref{lem: FinitelyManyCutsGeneral} for
  ${\hat{s}} = s_1$ and  $j = s_2$, we know that
  $\BR_i((x^*_{s_1})_{-i}) \neq \BR_i((x^*_{s_2})_{-i})$ has to
  hold since $i \in  N^\proxy(x^*_{s_2},\proxy^*_{s_2})$ is true by assumption.

  Since this holds for arbitrary sequence indices,
  $\BR_i((x^*_{s})_{-i}),s\in[j]$, must be pairwise different,
  implying the claim.
\end{proof}

With this lemma at hand, we are now in the position to prove
Theorem~\ref{thm: FinitelyManyCutsGeneral}.

\begin{proof}[Proof of Theorem~\ref{thm: FinitelyManyCutsGeneral}]
  Consider an arbitrary sequence of iterations of \Cref{algorithm:BC}
  with corresponding optimal solutions
  $(x^*_s,\lambda^*_s,\proxy^*_s,\piproxy^*_s)$, $s\in[l+1]$,
  and best responses $y_s^*$, $s\in[l]$, for an $l \in \N$.
  For every $s\leq l$, denote by $C_t^s$ the set $C_t$ defined via the
  cuts of node $t$ after the $(s-1)$-th iteration.
  Moreover, denote by $\Proj_{x,\proxy}(C_t^s)$ the projection of $C_t^s$ to
  the $(x,\proxy)$-space.
  Let $s_1^i<\ldots<s_{j_i}^i\leq l$ be the indices in which the
  feasible set was augmented with a cut from \Cref{lem: Cut} for
  player $i\in N$.
  Then,
  \begin{equation*}
    (x^*_{s_k^i},\proxy^*_{s_k^i}) \in \Proj_{x,\proxy}(C_t^{s_k^i})
    \quad \text{and} \quad
    i \in N^\proxy(x^*_{s_k^i}, \proxy^*_{s_k^i})
    \text{ for all } k\leq j_i.
  \end{equation*}
  Define $C^k_i$ as in \eqref{eq: FinitelyManyCutsGeneral}
  \wrt this subsequence and observe that $(x^*_{s_k^i},\proxy^*_{s_k^i})
    \in C^k_i$ holds for all $k \leq j_i$ since
    $x^*_{s_k^i} \in \FeasStrats$ and $C_t^{s_k^i}$ contains the
    cuts~\eqref{eq: PlayerCut} added for player~$i$ in the iterations
    $s_1^i,\dotsc,s_{k-1}^i$.  Thus,
  \Cref{lem: FinitelyManyCutsGeneral} is applicable, implying $j_i
  \leq \abs{\BRcom_i}$.
  Hence, $l =\sum_{i\in N} j_i \leq \sum_{i\in N} \abs{\BRcom_i}$,
  which shows the claim.
\end{proof}



\section{Adaptive B\&C for Best-Approximate Nash Equilibria}
\label{sec:best-approximate}
In this section, we brief\/ly discuss \aBinBC to find
best-approximate NE. Consider a player $i \in N$, a multiplicative approximation vector $\alpha_{-i}\in [1,\infty)^{n-1}$ for her rivals and an additive approximation vector $\beta \in \R^N_+$.
We describe in the following how to determine (up to some tolerance) the infimum~$\alpha_i^{\min}$ of all values $\alpha_i \geq 1$ such that an $(\alpha_i,\alpha_{-i},\beta)$-NE exists
(the computation of a best possible $\beta_i$ works analogously).
Note that the following explanation can be generalized to determining the minimal value $\gamma$ such that for a subset $I$ of players there exists a $(\gamma,\dotsc,\gamma,\alpha_{-I},\beta)$-NE.
Also note that via discretization of the remaining approximation values,
this allows to sample a Pareto-frontier of best-possible approximation vectors.

To provide a clearer intuition, let us first consider
a standard binary search over potential values of $\alpha_i^{\min}$
coupled with our B\&C method considered so far.
Starting with an initially given value~\mbox{$\alpha^+_i \geq 1$}, we check if
there exists an $(\alpha^+_i,\alpha_{-i},\beta)$-NE by using the described B\&C method
with Algorithm~\ref{algorithm:BC} to solve the nodes of the tree.
If so, we use this $\alpha^+_i$ as the upper bound.
Otherwise we set $\alpha^+_i \leftarrow F \alpha^+_i$ with an update
factor $F \gg 1$ and re-solve the problem.
This is carried out until we find an
appropriate $\alpha^+_i$ for which a corresponding approximate equilibrium
exists.
Note that this loop does not terminate if there is no
  $\alpha_i \geq 1$ such that an $(\alpha_i,\alpha_{-i},\beta)$-NE
  exists.
  Since best-response values are nonnegative, this is the case if and
  only if every $x \in \FeasStrats$ with
  $\pi_j(x) \leq \alpha_j \VBR_j(x_{-j}) + \beta_j$ for all $j \neq i$
  satisfies $\VBR_i(x_{-i}) = 0$ and $\pi_i(x) > \beta_i$.
After that, we carry out a binary search over the interval $[1,
\alpha^+_i]$ up to a given tolerance $\tau$, where in each step with
value $\tilde{\alpha}_i \in [1,\alpha^+_i]$,  we can apply our B\&C
method to determine the existence or non-existence of a corresponding
$(\tilde{\alpha}_i,\alpha_{-i},\beta)$-NE
and decrease or increase the value of $\tilde{\alpha}_i$ accordingly.
If the binary search stops with an interval $[1,\bar{\alpha}]$ with
$\bar{\alpha} \leq 1+\tau$, then the B\&C method is used again with
$\tilde{\alpha}_i=1$ to check explicitly for this value.

The above described simple binary search method is
a multi-tree algorithm since it explores a new B\&C tree in every
iteration. In the following, we describe how to realize a
single-tree implementation using a more sophisticated way to combine
the binary search with our B\&C method.
As described above, we determine an interval
$[1, \alpha^+_i]$ of potential values for $\alpha^{\min}_i$.
In what follows, let us call a node {\it explored} if
Algorithm~\ref{algorithm:BC} either branched on it or pruned it.
Moreover, we call a node {\it unexplored} if it is still in the queue
or if it is the node being currently processed.
In particular, if an approximate NE is found by
Algorithm~\ref{algorithm:BC} in a certain node of the tree, this node
is also considered unexplored.
As in the previously discussed approach, we perform a
binary search over $[1, \alpha^+_i]$, using our B\&C method
as a subroutine. However, if in some iteration with value
$\tilde{\alpha}_i$ we find a $(\tilde{\alpha}_i,\alpha_{-i},\beta)$-NE, we
do not need to restart the B\&C tree search from scratch for the next iteration.
Instead, we can recycle the set $\tilde{N}$ of unexplored
nodes corresponding to the iteration in which the $(\tilde{\alpha}_i,\alpha_{-i},\beta)$-NE was found.
Also, if at some point we prove there is no
$(\tilde{\alpha}_i,\alpha_{-i},\beta)$-NE, we can warm-start the search
of the next approximate NE from the tree search state of the last approximate NE found.
This requires storing the set of unexplored nodes (including
  their branching decisions and cuts) whenever an approximate NE is
  found and restoring it after an unsuccessful iteration.
  The nodes, cuts, and pruning decisions of an unsuccessful iteration
  with value~$\tilde{\alpha}_i$ have to be discarded since they are only
  guaranteed to be valid for $\tilde{\alpha}_i$, but not for the larger
  values considered afterward.
  The only exception are globally valid ANE-cuts
that are also valid independently of the approximate parameter values.
Indeed, they are satisfied by all points
  $(x,0,\VBR(x),\pi(x))$ with $x \in \FeasStrats$.
  The best-response cuts and the gradient cuts are
  such cuts. They can thus be reused in any other search
via Algorithm~\ref{algorithm:BC}.
This reuse of unexplored nodes and some cuts leads
to  a single-tree realization of the \BinBC where we
only have to update the constraints~\eqref{eq: LambdaBound}
involving the variable~$\lambda$ to fit to the new $\alpha_i$
parameter, as well as the cuts.
We consider three variants that differ in the amount of
recycled information: In \reusewithoutcuts, we only reuse the
tree structure, i.e., the branching and pruning decisions.
In \reusetreesearch, we also keep, in addition, the derived cuts
that are valid due to being derived in a B\&C call with a higher
value $\tilde{\alpha}$.
Finally, we denote \reusetreeandallcuts{} the
variant with the further
reuse of globally valid cuts that are independent of the approximation
parameters.

We now prove that these single-tree methods are correct.

\begin{proposition}
  \label{prop:adapt_cuts}
  Let $\alpha^1_i > \alpha^2_i \geq 1$ be given.
  Assume an $(\alpha^1_i,\alpha_{-i},\beta)$-NE was found using
  Algorithm~\ref{algorithm:BC} and ANE-cuts, and
  let $\tilde{N}$ be the respective set of unexplored nodes.
  Moreover, let $\hat{x}$ be an $(\alpha^2_i,\alpha_{-i},\beta)$-NE.
  Then, the point $\hat{z} \define (\hat{x}, 0, \VBR(\hat{x}),
  \pi(\hat{x}))$ is feasible for one of the nodes in $\tilde{N}$.
\end{proposition}
\begin{proof}
  Note that any $(\alpha^2_i,\alpha_{-i},\beta)$-NE is also an $(\alpha^1_i,\alpha_{-i},\beta)$-NE.
  Thus, since $\hat{x}$ is an $(\alpha^2_i,\alpha_{-i},\beta)$-NE, the point $\hat{z}$
  is feasible for \eqref{model:CpiHPR}, i.e., for the root-node
  problem of the tree search to compute an $(\alpha^1_i,\alpha_{-i},\beta)$-NE.

  There are three operations performed on nodes in
  Algorithm~\ref{algorithm:BC} in which the feasible point
  $\hat{z}$ could be removed:
  branching, cutting, and pruning.
  Since $\hat{z}$ satisfies all integrality constraints, no branching
  constraint can cut off $\hat{z}$.
  Again, because $\hat{x}$ is an $(\alpha^1_i,\alpha_{-i},\beta)$-NE as well,
  an ANE-cut for $\alpha^1_i$ that is added at a
    node~$t$ for which $\hat{z}$ is feasible does not cut off~$\hat{z}$
    by Condition~\ref{def:non-ne-cut:1} of \Cref{def:non-ne-cut} since
    $\hat{z} \in \Nespiapp \cap B_t \cap C_t$ holds for the
    approximation vectors $((\alpha^1_i,\alpha_{-i}),\beta)$.
  Finally, $\hat{z}$ is feasible and has an objective value of $0$, so
  it cannot be removed by pruning.
\end{proof}

From this it immediately follows that, provided every B\&C call
terminates, for a given valid upper
bound~$\tilde{\alpha}_i$, both variants \reusetreesearch{} and
\reusewithoutcuts{} compute $\alpha^{\min}_i$ (up
to the
tolerance for the interval size of the binary search).
To this end,
apply Proposition~\ref{prop:adapt_cuts} \wrt
 $\alpha^1_i$ set to
 the current best approximate
 value found, which exists because $\tilde{\alpha}_i$ is valid,
 and $\alpha^2_i$ set to any value smaller than $ \tilde{\alpha}_i$
 for which an $(\alpha^2_i,\alpha_{-i},\beta)$-NE exists (e.g.,
 $\alpha_i^2 \define \alpha^{\min}_i$ if the infimum is attained).
For NEPs, the conditions of \Cref{thm: FinitelyManyCutsGeneral}
ensure finite termination of these calls.
For GNEPs, the same holds for the fallback strategy under the
conditions stated in \Cref{rem: Fallback}.


\section{Numerical Results}
\label{sec:numerical-results}

In this section, we present numerical results on the computation
of the minimum approximation value $\alpha^{\min} \in \R$ such that an
$(\alpha^{\min},\beta)$-NE exists.
Here, we restrict ourselves again to the case in which every player has the same
approximation factor $\alpha^{\min}$ and slightly abuse notation in writing
$(\alpha^{\min},\beta)$-NE instead of $((\alpha^{\min},\ldots,\alpha^{\min}),(\beta,\ldots,\beta))$-NE.
We present the considered games in
Section~\ref{sec:test-instances}, as well as a description of the
generation of instances in Section \ref{subsec:generation_instances}.
Then, we discuss the implementation details as well as the software and
hardware setup in Section~\ref{sec:impl-deta}.
Moreover, the numerical results are discussed in
Section~\ref{sec:analysis-of-the-results}. The Section~\ref{sec:pareto-front}
shows Pareto-like fronts of the approximation parameter for two NEP
implementation game instances.

\subsection{Game Descriptions}
\label{sec:test-instances}
\subsubsection{Implementation Games}

We study a  model of \textcite{Kelly98}  in the domain of TCP-based
congestion control.
To this end, we consider a directed graph~$G=(V,E)$ with nodes~$V$ and
edges~$E$.
The set of players is given by $N= \set{1, \dots, n}$ and each
player~$i\in N$ is associated with an end-to-end pair $(s_i,t_i)\in
V\times V$.
The strategy~$x_i$ of player~$i \in N$ represents an integral
$(s_i, t_i)$-flow with a flow value equal to her demand $d_i \in
\nonNegInts$.
Thus, the strategy set of a player $i \in N$ is described by
\begin{align}
  \label{eq: DFGStrat}
  X_i\define \defset{ x_i\in \Z_{\geq 0}^E}{ A_Gx_i = b_i,\ x_i \leq c} \cup
  \set{0},
\end{align}
i.e., the union of the $0$-flow and the flow
polyhedron of player $i$ with $A_G$ being the arc-incidence matrix of
the graph~$G$, $c \in \R^E_{\geq 0}$ being a vector of edge
capacities that bound the flow of every player individually, and~$b_i$ being the vector with $(b_i)_{s_i} = d_i$,
$(b_i)_{t_i} = -d_i$, and $0$ otherwise.
Note that this allows players to not participate in the game because
$x_i
= 0$ is a feasible strategy.
All players want to maximize their utility given by $\mu_i^\top x_i$
for player $i$ choosing strategy~$x_i$ for a given vector $\mu_i\in
\R^E_{\geq 0}$.

In addition to the set~$N$ of players, there is a
central authority, which aims to determine a price vector~$p^*\in \R_{\geq
  0}^E$
for the edges with the goal to (weakly) \emph{implement} a certain
edge load vector $u \in \R^E_{\geq 0}$, i.e., the authority wants to
determine a price vector $p^*$ such that there exists a strategy
profile~$x^*$ of the players in~$N$ with the following properties.
\begin{enumerate}[label=(\roman*)]
\item \label{impl: 1} The load is at most~$u$, i.e.,
  $\ell(x^*)\define \sum_{i \in N} x_i^* \leq u$.
\item \label{impl: 2} The strategy~$x^*$ is an equilibrium for the
  given~$p^*$,
  i.e.,
  \begin{equation*}
    x_i^* \in \argmax \Defset{(\mu_i - p^*)^\top x_i}{x_i \in X_i}
  \end{equation*}
  holds for all $i \in N$.
\item \label{impl: 3} The edges for which the targeted load is not
  fully used have zero price, i.e., $\ell_e(x^*) < u_e$ implies $p_e^*
  =0$,
\item \label{impl: 4} The price is bounded from above, i.e., $p^*\leq
  p^{\max}$.
\end{enumerate}
Here, $p^{\max}\in \R^E_{\geq 0}$ is some upper bound on the prices
satisfying
\begin{equation*}
  p^{\max}_e > \abs{E} \cdot \max_{e' \in E}(\mu_{i})_{e'} \cdot
  \max_{e' \in E}c_{e'}
\end{equation*}
for all $i \in N$ and $e \in E$.

For the setting in which players are allowed to send fractional arbitrary amounts of flow,
\textcite{Kelly98}
proved that every vector $u$ is weakly implementable.
Allowing a fully fractional distribution of the flow, however, is not
possible in some applications---the notion of data packets as
indivisible units seems more realistic.
The issue of completely fractional routing versus integrality
requirements has been explicitly addressed by \textcite{Orda93},
\textcite{HarksK16b}, and \textcite{wang2011}. More recently,
\textcite{HarksSchwarzPricing} introduced a unifying framework for
pricing in nonconvex resource allocation games,
which, in particular, encompasses the integrality-constrained version
of the model
originally studied by \textcite{Kelly98}.
In a related model with variable flow amounts, they proved
(Corollary~7.8) that every nonnegative integral capacity vector is
weakly implementable when all players share a common source and have
identical non-decreasing affine bandwidth utilities.
This result does not directly apply to the fixed-demand model
considered here.
However, in the general case, the implementability of a vector $u$ is
not guaranteed.
This raises the question of which vectors are implementable and which
are not.

We can model this question as a NEP with $n+1$ players in which the
first $n$~players correspond to the player set~$N$ and the $(n+1)$-th
player is the central authority.
We denote by $(x,p)$ a strategy  profile and set the costs
to the negated utility $\pi_i(x_i,x_{-i},p) = (p-\mu_i)^\top  x_i$
for $i \in N$ and the costs of the central authority
to  $\pi_{n+1}(p,x) = (u - \ell(x))^\top p$.
The strategy spaces are given by $X_i$ in \eqref{eq: DFGStrat} for $i
\in N$ and by $X_{n+1}\define\defset{p \in \R ^E_{\geq 0}}{p\leq p^{\max}}$.

The same argument as in Lemma~5.1 of \textcite{DSHS25} shows that a tuple
$(x^*,p^*)$ weakly implements~$u$ if and only if $(x^*,p^*)$ is an
exact equilibrium of the constructed NEP.
Note that in \textcite{DSHS25}, the authors consider
a jointly capacity-constrained version of this problem,
leading to a GNEP reformulation. Yet, for sufficiently large
capacities, the problems become equivalent.
In the computational study of \textcite{DSHS25}, the authors proved
the non-existence of implementing prices for several instances.
Here, we now apply our methods to compute a minimal multiplicative
factor $\alpha^{\min}$ such that a corresponding approximate
equilibrium exists for the NEP versions of these instances.

\subsubsection{Generalized Quadratic Concave  Games}
\label{generalized_quadratic_concave_game_description}

We study a GNEP with $n$ players, only integer variables and
quadratic concave cost functions for every player.
The best-response problem of each player $i$ is given by
\begin{align*}
  \min_{x_i \in \mathbb{R}^m} \quad & \sum_{j=1}^m -c_i^j {x_i^j}^2 \\
  \st \quad
                                    & A_i^i x_i + A_i^{-i} x_{-i} \leq b_i, \\
                                    & 0 \leq x_i \leq 10, \quad x_i \in \Z^m.
\end{align*}
where the parameters $c_i \in \mathbb{R}^m$, $A_i^i \in
\mathbb{R}^{q \times m}$, $A_i^{-i} \in \mathbb{R}^{q
  \times m(n-1)}$, and $b_i \in \mathbb{R}^{q}$ are all
nonnegative. This allows the feasible strategy profiles to be nonempty
because the strategy profile with only zeros is feasible. The
\Nefreeset \wrt $\piproxy$ used is $\freeset[\piproxy,\conc]_{i}$
because Assumption~\ref{ass: IC}\ref{ass: IC2: conc} is
satisfied.

\subsection{Generation of Instances} \label{subsec:generation_instances}

We use the 450 implementation game instances from \textcite{DSHS25}.
To obtain NEP versions of these instances, we replaced the
  joint capacity restriction $\sum_{i \in N} x_i \leq c$ by the
  individual restrictions $x_i \leq c$, $i \in N$; see~\eqref{eq:
    DFGStrat}.

We create instances for the generalized quadratic concave games in the
following way. For number of players $n \in \{2,3,4\}$, number of
variables per player $m \in \{ 3,6,10\}$, number of
constraints per player's optimization problem $q = 5$, proportion of
nonzero coefficients $f \in \{0.1,0.25,0.5\}$ and a number of
repetitions $r=10$, we randomly create instances named
$\mathrm{randomGNEP}\{n\}-\{m\}-\{q\}-\{100*f\}-\{r\}$, where $\{x\}$ is
replaced by the value of $x$. The $j$-th cost coefficient $c_i^j$ of
player $i$ is chosen uniformly
in $\{1,\ldots,10\}$. Coefficients in $A_i^i$ or in
$A_i^{-i}$ are
chosen as nonzeros with probability $f$ and with value uniformly in
$\{1,\ldots,10\}$, and if for one constraint all own
player's
coefficients are zeros, one of them is selected randomly and changed
to a nonzero value
uniformly in $\{1,\ldots,10\}$ so that each constraint
depends on the
player's strategy. Finally, the right-hand side coefficients in $b_i$
are chosen uniformly in $\{10,\ldots,100\}$.

\subsection{Implementation Details}
\label{sec:impl-deta}

The computations have been executed on a single core
Intel Xeon Gold 6126 processor at \SI{2.6}{GHz} with \SI{17}{GB} of
RAM for the NEP implementation games, and \SI{4}{GB} for the
quadratic concave generalized games. 
The time limit for each run of the
whole adaptive B\&C is 3600 seconds.
The code is implemented in \textsf{C++} and compiled with
\textsf{GCC}~13.1.

In Lines~\ref{line:step_1} and~\ref{line: SolveBR} of Algorithm~\ref{algorithm:BC},
we solve MIQCPs or MIQPs using \textsf{Gurobi}~13.0.1
\parencite{gurobi} with the parameter \textsf{numericFocus} set to 3,
and the parameter \textsf{MIPGap} set to its default when
solving the node problem and set to $0$ when solving the best-response
problems.
Also, the parameter \textsf{SoftMemLimit} of Gurobi is set to
\SI{16}{GB} (respectively \SI{3}{GB}) for the solver \textsf{Gurobi}
for the implementation games (resp.\ for the generalized quadratic
concave games).
Moreover, the cuts in Line \ref{line: Cut} are added in the following
way. For each player $i$, a cut to strengthen the approximation of 
$\pi_i(x)$ via $\piproxy_i$ is computed 
if $\piproxy^*_i < \pi_i(x^*) - \delta$ (only for GNEPs, because
the NEP node model directly use $\pi_i(x)$ instead of $\piproxy_i$ 
as explained in Section~\ref{sec:cuts-ne}) and
if not, a cut to strengthen the approximation of $\VBR_i(x_{-i})$
via $\proxy_i$ if $\proxy^*_i > \VBR_i(x_{-i}^*) + \delta$ with
$\delta=10^{-6}$. The resulting cuts are only added to the model
if the violation of the cut evaluated at the
optimal solution of the current
node~$(x^*,\lambda^*,\proxy^*,\piproxy^*)$ is greater than
$5 \cdot 10^{-6}$. If no cuts are added, 
the algorithm stops with an error indicating numerical issues.
All the non-default parameter values have been chosen based on
preliminary numerical testing.

In both games, the best-response values are
nonpositive. Indeed, in the NEP implementation games every player can choose
the zero flow and the central authority the zero price vector, both
with cost~$0$.
Hence, we use the constraints~\eqref{eq: LambdaBound} in the form
of Remark~\ref{rem:NE_definition}, and consider a strategy
profile~$x$ to be an $(\alpha,\beta)$-NE if
\begin{equation*}
  \alpha_i \bigl(\pi_i(x) - \beta_i\bigr) - \Phi_i(x_{-i}) \leq 10^{-5}
  \quad \text{for all } i \in N,
\end{equation*}
i.e., if these constraints evaluated at $\lambda = 0$,
$\proxy = \Phi(x)$, and $\piproxy = \pi(x)$ are satisfied up to
$10^{-5}$.
For the pruning step in Line~\ref{line: Infeasible}, we check
if the objective value is greater than $10^{-5}$.
With this setting, the check for an approximate NE and the pruning
step use the same tolerance on the same scale.
The bounds of the node problems are computed from the
following continuous optimization problems:
\begin{equation*}
  \proxy_i^+ = \sup \Defset { \pi_i(x) }{x \in
    \FeasStratsRel}
  \quad \text{and} \quad
  \piproxy_i^- = \inf \Defset { \pi_i(x) }{x \in \FeasStratsRel}.
\end{equation*}%
Regarding the starting approximate parameters $(\alpha,\beta)$ of the
\minapprox and its variants, preliminary results showed that proving
the potential non-existence of $(1,\beta)$-NEs, i.e., best possible
approximate NE with $\beta$ fixed, is really expensive. Thus, we
prefer to start with a middle value $\alpha^+_i=10$ for each player to
concentrate the computation time to reducing the bounds on the best
$\alpha$. Moreover, we set the interval size tolerance
to $0.01$, and the factor $F$ to 10.

The exploration strategy of the branching scheme is depth-first
search, while the variable chosen for branching is the most
fractional one. In case of a tie, the smallest index is chosen.
While the performance of our method most likely would benefit from
more sophisticated node selection strategies and branching rules,
their study and implementation is out of the scope of this paper.

\subsection{Discussion of the Results}
\label{sec:analysis-of-the-results}

We compare the relative efficiency of all applicable variants of the adaptive B\&C method on NEP implementation games in Section \ref{subsubsec:NEP_implementation_game_results} and GNEPs with the quadratic concave generalized games in Section \ref{subsubsec:quadratic_concave_generalized_game_results}.
Let us emphasize that when we speak in the following of ''solving an instance'',
we mean that the respective algorithm found (up to tolerance)
the best possible approximation value $\alpha$ for some given $\beta$.
Note that this is
significantly harder than checking if an
$(\tilde{\alpha},\beta)$-NE exist, because Algorithm~\ref{algorithm:BC} has
to be called multiple times. Another reason is that, the closer $\alpha$ is
to the best possible value, the harder it is in general to find such an
approximate NE or prove that none exist.

\subsubsection{Implementation Games} \label{subsubsec:NEP_implementation_game_results}

We look for the minimum multiplicative parameter $\alpha$ such that an
$((\alpha,\ldots,\alpha,1),(1,\ldots,1,0))$-NE exists, i.e., NEs
where the multiplicative (respectively additive) approximation
parameter $\alpha \in [1,\infty)$ (resp.\ $\beta = 1$) is
the same for each player except for the last player, the central authority, which plays
exact best responses. The latter ensures that, in an approximate equilibrium,
prices are $0$ on all resources with slack.
By \Cref{prop:adapt_cuts}, the simple binary search as well as the variants
\reusetreesearch{} and \reusewithoutcuts{} of the
adaptive B\&C can be used. In addition, the variant
\reusetreeandallcuts{} can also be used
because the equilibrium cuts are globally valid and independent of any approximation
parameter value $(\alpha,\beta)$.

Figure~\ref{figure:ECDF_implementationGame_characteristics} shows the
results on the NEP implementation games. Among the 450 instances, 55 (\SI{12.2}{\%})
were solved by the \minapprox[,] 73 (\SI{16.2}{\%}) by the variant
\reusetreesearch{}, 71 (\SI{15.8}{\%}) by the variant \reusewithoutcuts{}, and
70 (\SI{15.6}{\%}) by the variant \reusetreeandallcuts{} of
our \BinBC[.]
The left figure shows ECDFs of instances solved with respect to the computation time for every variant. It shows that the \minapprox is less
efficient than the three other variants to find the best
multiplicative parameter up to tolerance on those instances,
while the three variants perform similarly.
Further experiments, which are beyond the scope of this work, may reveal their respective merits.
The right plot shows an ECDF of the best $\alpha$ values found
for the variant \reusetreesearch{}. An approximate NE
with $\alpha=1$ is found in \SI{12.9}{\%} of the instances.
The sharp change at $\alpha=10$ is an artifact caused by the
starting value $\alpha^+$ having this value.
Indeed, the changes of the curve for $\alpha>10$ indicate that
the first B\&C call proved nonexistence, while the changes for
$\alpha \leq 10$ indicate that an approximate NE was found in
the first call.
It is expected that it is harder to show nonexistence compared to finding an approximate
NE because in the first case you have to explore the entire branching tree while
in the second you just have to find one point.
For this instance set, variant \reusetreesearch{} only proved 
nonexistence with $\alpha=10$ for three instances, proved existence for 
200 instances, while 247 instances
(\SI{54.9}{\%}) stopped before the end of the first call to the B\&C.

\begin{figure}
  \includegraphics[width=0.5\linewidth]{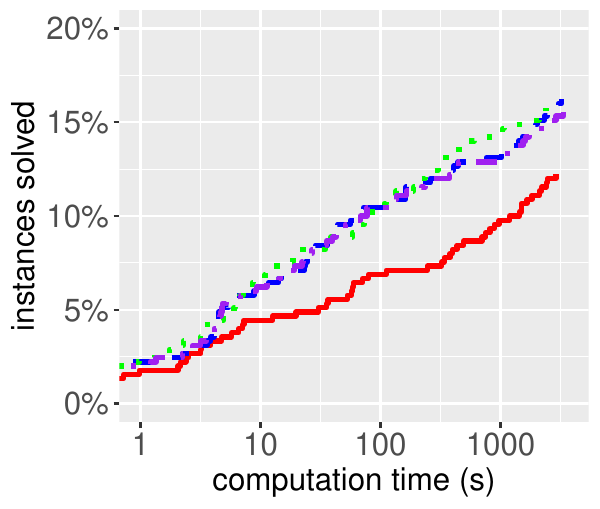}%
  \includegraphics[width=0.5\linewidth]{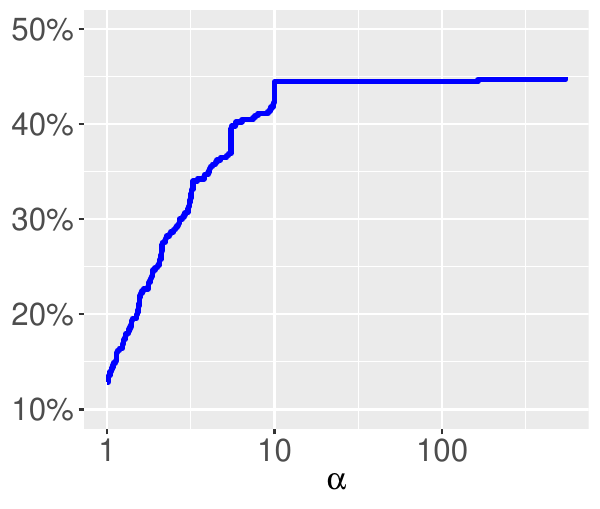}
  \caption{Left: ECDF of instances solved with respect to the
    computation time for the \minapprox
    (solid red line), the \reusetreesearch{} variant (dashed blue line),
    the \reusewithoutcuts{} variant (dotted green line), and the
    \reusetreeandallcuts{} variant (dash-dotted purple line).
    Right: Proportion of instances for which an $(\alpha,1$)-NE or better
    was found using the \reusetreesearch{} variant.}
  \label{figure:ECDF_implementationGame_characteristics}
\end{figure}

\subsubsection{Quadratic Concave Generalized Games} \label{subsubsec:quadratic_concave_generalized_game_results}

We look for the minimum multiplicative parameter $\alpha$ such that an $(\alpha,1)$-NE exists.
By \Cref{prop:adapt_cuts}, the simple binary search as well as the variants
\reusetreesearch{} and \reusewithoutcuts{} of the adaptive B\&C can be used.

Among the 270 instances, the \minapprox solves 96 instances (\SI{35.6}{\%}),
the variant \reusetreesearch{} 101 (\SI{37.4}{\%}) and the variant
\reusewithoutcuts{} 103 (\SI{38.1}{\%}). Each variant shows that 107 instances
(\SI{39.6}{\%}) have no $(1,1)$-NE. In addition, each variant finds an
($\alpha,1$)-NE with a finite $\alpha$ for 236 instances (\SI{87.4}{\%}).
Figure \ref{figure:ECDF_quadraticConcave_characteristics} shows two plots.
The left plot shows the number of instances solved according to the computation
time for each instance. The \minapprox seems slightly less efficient
to find the best multiplicative parameter $\alpha$ up to the tolerance.
The right plot shows an ECDF of the best values of $\alpha$ found for which there exists
an $(\alpha,1)$-NE for the variant \reusewithoutcuts{}. The sharp turn for $\alpha=\alpha^+=10$ has the same interpretation as for the NEP
implementation games.  In particular, 26
instances (\SI{9.6}{\%}) prove that $\alpha^*>10$.

\begin{figure}
  \includegraphics[width=0.5\linewidth]{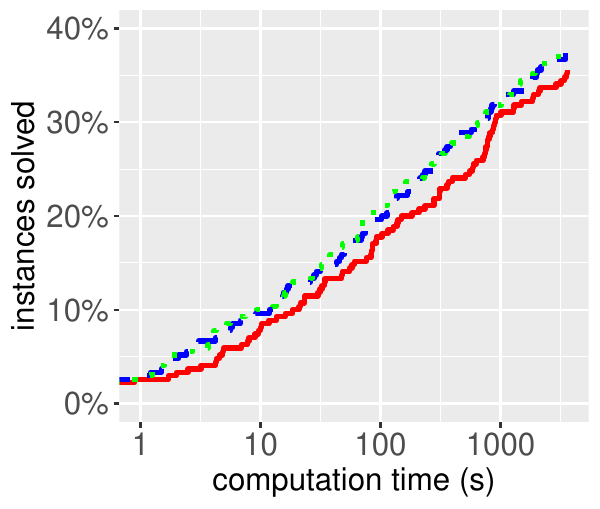}%
  \includegraphics[width=0.5\linewidth]{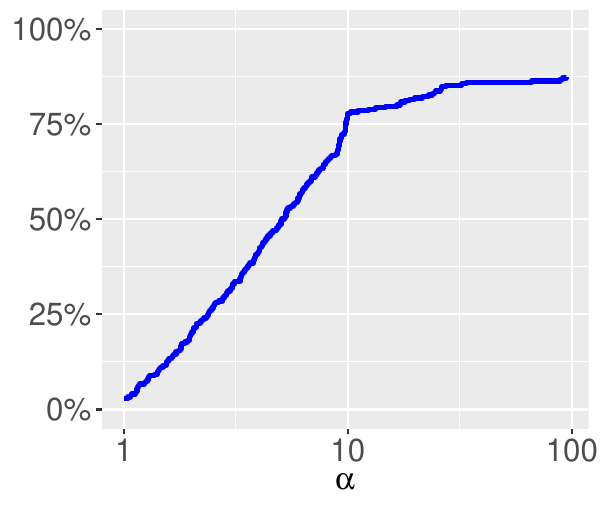}
  \caption{Left: ECDF of instances solved with respect to the
    computation time for the \minapprox
    (solid red line), the \reusetreesearch{} variant (dashed blue line),
    and the \reusewithoutcuts{} variant (dotted green line).
    Right: Proportion of instances for which an
    $(\alpha,1$)-NE or better was found for the variant
    \reusewithoutcuts{}.}
  \label{figure:ECDF_quadraticConcave_characteristics}
\end{figure}


\subsection{\texorpdfstring{Minimum Multiplicative Factors for Prescribed Additive Parameters in}{Minimum Multiplicative Factors for Prescribed Additive Parameters in} NEP Implementation Games}
\label{sec:pareto-front}

Using the algorithm described in Section \ref{sec:best-approximate}, 
we approximate the infimum $\alpha$ for prescribed $\beta$ for instances of the NEP implementation game:
Figure~\ref{figure:pareto_front} shows two curves of the infimum $\alpha$
up to a tolerance 0.01 for $((\alpha,\ldots,\alpha,1),(\beta,\ldots,\beta,0))$-NE.
They were built by discretizing the value of 
$\beta$ and minimizing $\alpha$ such that there exists an 
$((\alpha,\ldots,\alpha,1),(\beta,\ldots,\beta,0))$-NE for the 
corresponding $\beta$.
The discretization in $\beta$ is $0,3,\ldots,30$ for instance 
I\_2\_10\_mm\_10\_10 while it is $0,0.4,\ldots,4$ for instance 
I\_2\_10\_ss\_10\_7.

\begin{figure}
  \includegraphics[width=0.5\linewidth]{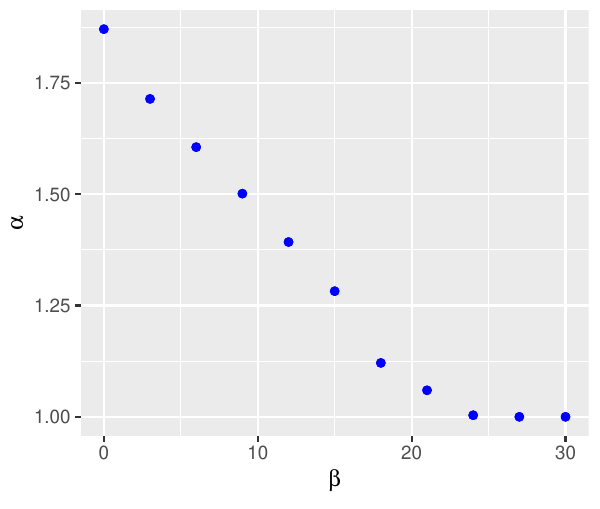}%
  \includegraphics[width=0.5\linewidth]{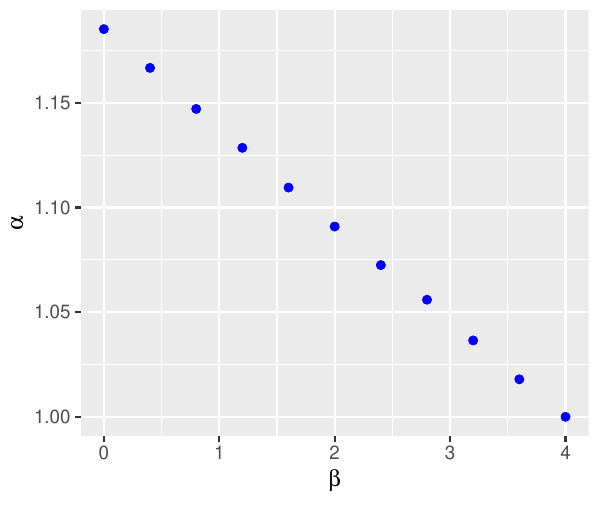}
  \caption{Infimum $\alpha$ (up to a tolerance of 0.01) such that an
$((\alpha,\ldots,\alpha,1),(\beta,\ldots,\beta,0))$-NE exists for two 
instances of NEP implementation games: I\_2\_10\_mm\_10\_10 to the 
left and I\_2\_10\_ss\_10\_7 to the right. 
A blue point corresponds to the best upper bound when an approximate 
value for the best $\alpha$ was proved up to a tolerance of $0.01$.}
  \label{figure:pareto_front}
\end{figure}

\section{Conclusion}
\label{sec:conclusion}

We presented a B\&C method for computing
$(\mappro,\addappro)$-NE for standard and generalized Nash equilibrium
problems with mixed-integer variables.
For the GNEP case, the method relies on the existence of suitable cuts,
which we derive under appropriate assumptions using intersection cuts.
For the special case of NEPs, we consider a different type of cut
and show that our method terminates in finite time provided that
each player has only finitely many distinct best-response sets.
Based on these cuts, the fallback strategy of \textcite{DSHS25},
which fixes recurring integer profiles, carries over to our setting and
ensures finite termination in the GNEP case under analogous
assumptions.
Building upon this B\&C approach, we further introduce a single-tree
binary-search method to compute best-approximate equilibria under some
simplifying assumptions.
A first numerical case study for a class of mixed-integer flow games
and for GNEPs with concave quadratic cost functions shows the applicability of the approach.

Future work based on this contribution includes studying weaker
assumptions for deriving intersection cuts; e.g., we believe that the
assumption of affine constraints can be relaxed to constraints being
convex.


\section*{Acknowledgements}

This research has been funded by the Deutsche Forschungsgemeinschaft
(DFG) in the project 543678993 (Aggregative gemischt-ganzzahlige
Gleichgewichtsprobleme: Existenz, Approximation und Algorithmen).
We acknowledge the support of the DFG.
The computations were executed on the high performance cluster
``Elwetritsch'' at the TU Kaiserslautern, which is part of the
``Alliance of High Performance Computing Rheinland-Pfalz'' (AHRP).
We kindly acknowledge the support of RHRK.

\section*{Use of AI}

During the preparation of this manuscript, we used Open AI's Codex
  and Anthropic's Claude to support the drafting and revision of parts
  of the text. Moreover, we let these AI assistants check the overall
  correctness of the presented theory.
  We carefully reviewed all resulting content and take full
  responsibility for the paper.


\printbibliography

@book{Beckmann56,
  author    = {Beckmann, Martin J. and McGuire, C. B. and Winsten, Christopher B.},
  location  = {New Haven, Connecticut},
  publisher = {Yale University Press},
  date      = {1956},
  title     = {Studies in the Economics of Transportation},
  url       = {https://cowles.yale.edu/sites/default/files/2022-08/specpub-beckmann-mcguire-winsten.pdf},
}

@article{Rubinstein17,
  author       = {Rubinstein, Aviad},
  title        = {Settling the complexity of computing approximate two-player {Nash} equilibria},
  year         = {2017},
  publisher    = {Association for Computing Machinery},
  address      = {New York, NY, USA},
  volume       = {15},
  number       = {2},
  pages        = {45--49},
  journaltitle = {ACM SIGecom Exchanges},
  doi          = {10.1145/3055589.3055596},
}

@article{dragotto2023zero,
  author       = {Dragotto, Gabriele and Scatamacchia, Rosario},
  publisher    = {INFORMS},
  date         = {2023},
  doi          = {10.1287/ijoc.2022.0282},
  journaltitle = {INFORMS Journal on Computing},
  number       = {5},
  pages        = {1143--1160},
  title        = {The zero regrets algorithm: Optimizing over pure {Nash} equilibria via integer programming},
  volume       = {35},
}

@article{Sagratella2017AlgorithmsFG,
  author       = {Sagratella, Simone},
  date         = {2017},
  doi          = {10.1007/s10589-017-9927-4},
  journaltitle = {Computational Optimization and Applications},
  pages        = {689--717},
  title        = {Algorithms for generalized potential games with mixed-integer variables},
  volume       = {68},
  number       = {3},
}

@article{Sagratella19,
  author       = {Sagratella, Simone},
  publisher    = {Taylor \& Francis},
  date         = {2019},
  doi          = {10.1080/02331934.2018.1545125},
  journaltitle = {Optimization},
  number       = {1},
  pages        = {197--226},
  title        = {On generalized {Nash} equilibrium problems with linear coupling constraints and mixed-integer variables},
  volume       = {68},
}

@book{Conforti-et-al:2014,
  author    = {Conforti, Michele and Cornuéjols, Gérard and Zambelli, Giacomo},
  location  = {Cham},
  publisher = {Springer International Publishing},
  date      = {2014},
  doi       = {10.1007/978-3-319-11008-0},
  title     = {Integer Programming},
  isbn      = {9783319110073},
  series    = {Graduate Texts in Mathematics},
  number    = {271},
}

@article{balas1971intersection,
  author       = {Balas, Egon},
  publisher    = {INFORMS},
  date         = {1971},
  doi          = {10.1287/opre.19.1.19},
  journaltitle = {Operations Research},
  number       = {1},
  pages        = {19--39},
  title        = {Intersection cuts—a new type of cutting planes for integer programming},
  volume       = {19},
}

@article{Anderson13,
  author       = {Anderson, Edward J.},
  date         = {2013},
  doi          = {10.1007/s10107-013-0691-7},
  journaltitle = {Mathematical Programming},
  number       = {2},
  pages        = {323--349},
  title        = {On the existence of supply function equilibria},
  volume       = {140},
}

@article{Koeppe11,
  author       = {Köppe, Matthias and Ryan, Christopher Thomas and Queyranne, Maurice},
  publisher    = {Institute for Operations Research and the Management Sciences (INFORMS)},
  date         = {2011},
  doi          = {10.1287/opre.1110.0964},
  issn         = {1526-5463},
  journaltitle = {Operations Research},
  number       = {6},
  pages        = {1445--1460},
  title        = {Rational Generating Functions and Integer Programming Games},
  volume       = {59},
}

@incollection{carvalho2023integer,
  author    = {Carvalho, Margarida and Dragotto, Gabriele and Lodi, Andrea and Sankaranarayanan, Sriram},
  publisher = {INFORMS},
  booktitle = {Tutorials in Operations Research: Advancing the Frontiers of OR/MS: From Methodologies to Applications},
  date      = {2023},
  doi       = {10.1287/educ.2023.0260},
  pages     = {31--51},
  title     = {Integer Programming Games: A Gentle Computational Overview},
  editor    = {Bish, Ebru and Balasubramanian, Hari and Shier, Douglas and Greenberg, Harvey J.},
  isbn      = {9798988285618},
}

@article{cronert2022equilibrium,
  author       = {Crönert, Tobias and Minner, Stefan},
  publisher    = {INFORMS},
  date         = {2024},
  doi          = {10.1287/opre.2022.2413},
  journaltitle = {Operations Research},
  number       = {2},
  pages        = {816--831},
  title        = {Equilibrium identification and selection in finite games},
  volume       = {72},
}

@inproceedings{PiaFM17,
  author    = {Del Pia, Alberto and Ferris, Michael and Michini, Carla},
  booktitle = {Proceedings of the 2017 Annual ACM-SIAM Symposium on Discrete Algorithms (SODA)},
  date      = {2017},
  doi       = {10.1137/1.9781611974782.37},
  pages     = {577--588},
  title     = {Totally Unimodular Congestion Games},
  publisher = {Society for Industrial and Applied Mathematics},
  editor    = {Klein, Philip N.},
  isbn      = {9781611974782},
}

@inproceedings{KleerS17,
  author    = {Kleer, Pieter and Schäfer, Guido},
  location  = {New York, NY, USA},
  venue     = {Cambridge, Massachusetts, USA},
  publisher = {Association for Computing Machinery},
  booktitle = {Proceedings of the 2017 ACM Conference on Economics and Computation},
  date      = {2017},
  doi       = {10.1145/3033274.3085149},
  isbn      = {9781450345279},
  pages     = {223--240},
  series    = {EC '17},
  title     = {Potential Function Minimizers of Combinatorial Congestion Games: Efficiency and Computation},
}

@article{guo2021copositive,
  author  = {Guo, Cheng and Bodur, Merve and Taylor, Joshua A.},
  title   = {Copositive Duality for Discrete Energy Markets},
  journal = {Management Science},
  doi     = {10.1287/mnsc.2023.00906},
  date    = {2026},
  volume  = {72},
  number  = {3},
  pages   = {2022--2040},
}

@article{kirst2022branch,
  author       = {Kirst, Peter and Schwarze, Stefan and Stein, Oliver},
  date         = {2024},
  doi          = {10.1137/23M1548189},
  journaltitle = {SIAM Journal on Optimization},
  number       = {4},
  pages        = {3371--3398},
  title        = {A Branch-and-Bound Algorithm for Nonconvex {Nash} Equilibrium Problems},
  volume       = {34},
}

@article{carvalho2021cut,
  author       = {Carvalho, Margarida and Dragotto, Gabriele and Lodi, Andrea and Sankaranarayanan, Sriram},
  date         = {2026},
  journaltitle = {Operations Research},
  title        = {The Cut-and-Play Algorithm: Computing {Nash} Equilibria via Outer Approximations},
  volume       = {74},
  doi          = {10.1287/opre.2023.0327},
  number       = {2},
  pages        = {1070--1086},
}

@article{schwarze2023branch,
  author       = {Schwarze, Stefan and Stein, Oliver},
  publisher    = {Springer},
  date         = {2023},
  doi          = {10.1007/s10589-023-00500-4},
  journaltitle = {Computational Optimization and Applications},
  number       = {2},
  pages        = {491--519},
  title        = {A branch-and-prune algorithm for discrete {Nash} equilibrium problems},
  volume       = {86},
}

@article{Sagratella16,
  author       = {Sagratella, Simone},
  date         = {2016},
  doi          = {10.1137/15M1052445},
  journaltitle = {SIAM Journal on Optimization},
  number       = {4},
  pages        = {2190--2218},
  title        = {Computing All Solutions of {{Nash}} Equilibrium Problems with Discrete Strategy Sets},
  volume       = {26},
}

@article{carvalho2022computing,
  author       = {Carvalho, Margarida and Lodi, Andrea and Pedroso, João Pedro},
  date         = {2022},
  doi          = {10.1016/j.ejor.2022.03.048},
  issn         = {0377-2217},
  journaltitle = {European Journal of Operational Research},
  number       = {3},
  pages        = {1057--1070},
  title        = {Computing equilibria for integer programming games},
  volume       = {303},
}

@article{Kelly98,
  author       = {Kelly, Frank P. and Maulloo, A. K. and Tan, D. K. H.},
  date         = {1998},
  doi          = {10.1057/palgrave.jors.2600523},
  journaltitle = {Journal of the Operational Research Society},
  number       = {3},
  pages        = {237--252},
  title        = {Rate control for communication networks: Shadow prices, proportional fairness and stability},
  volume       = {49},
}

@article{Harks24,
  author       = {Harks, Tobias and Schwarz, Julian},
  date         = {2025},
  doi          = {10.1007/s10107-024-02063-6},
  journaltitle = {Mathematical Programming},
  pages        = {231--277},
  title        = {Generalized {Nash} equilibrium problems with mixed-integer variables},
  volume       = {209},
  number       = {1--2},
}

@manual{gurobi,
  author  = {{Gurobi Optimization, LLC}},
  url     = {https://docs.gurobi.com/projects/optimizer/en/13.0/},
  date    = {2026},
  title   = {{Gurobi Optimizer Reference Manual}},
  version = {13.0},
  urldate = {2026-11-18},
}

@article{fischetti2018use,
  author       = {Fischetti, Matteo and Ljubić, Ivana and Monaci, Michele and Sinnl, Markus},
  date         = {2018},
  doi          = {10.1007/s10107-017-1189-5},
  journaltitle = {Mathematical Programming},
  number       = {1--2},
  pages        = {77--103},
  title        = {On the use of intersection cuts for bilevel optimization},
  volume       = {172},
}

@article{Fabiani_Grammatico:2020,
  author       = {Fabiani, Filippo and Grammatico, Sergio},
  date         = {2020},
  doi          = {10.1109/TITS.2019.2901505},
  journaltitle = {IEEE Transactions on Intelligent Transportation Systems},
  number       = {3},
  pages        = {1064--1073},
  title        = {Multi-Vehicle Automated Driving as a Generalized Mixed-Integer Potential Game},
  volume       = {21},
}

@article{Lozano_Smith:2017,
  author       = {Lozano, Leonardo and Smith, J. Cole},
  date         = {2017},
  doi          = {10.1287/opre.2017.1589},
  journaltitle = {Operations Research},
  number       = {3},
  pages        = {768--786},
  title        = {A Value-Function-Based Exact Approach for the Bilevel Mixed-Integer Programming Problem},
  volume       = {65},
}

@report{Horländer_et_al:2024,
  author      = {Horländer, Andreas and Ljubić, Ivana and Schmidt, Martin},
  url         = {https://optimization-online.org/?p=25955},
  date        = {2024},
  title       = {Using Disjunctive Cuts in a Branch-and-Cut Method to Solve Convex Integer Nonlinear Bilevel Problems},
  type        = {Preprint},
  institution = {Optimization Online},
}

@inproceedings{Fabiani_et_al:2022,
  author    = {Fabiani, Filippo and Franci, Barbara and Sagratella, Simone and Schmidt, Martin and Staudigl, Mathias},
  booktitle = {2022 IEEE 61st Conference on Decision and Control (CDC)},
  date      = {2022},
  doi       = {10.1109/CDC51059.2022.9993250},
  pages     = {4137--4142},
  title     = {Proximal-like algorithms for equilibrium seeking in mixed-integer {Nash} equilibrium problems},
  publisher = {IEEE},
}

@article{Orda93,
  author       = {Orda, Ariel and Rom, Raphael and Shimkin, Nahum},
  publisher    = {IEEE},
  date         = {1993},
  doi          = {10.1109/90.251910},
  journaltitle = {IEEE/ACM Transactions on Networking},
  number       = {5},
  pages        = {510--521},
  title        = {Competitive routing in multiuser communication networks},
  volume       = {1},
}

@article{HarksK16b,
  author    = {Harks, Tobias and Klimm, Max},
  date      = {2016},
  title     = {Congestion games with variable demands},
  number    = {1},
  pages     = {255--277},
  volume    = {41},
  journal   = {Mathematics of Operations Research},
  publisher = {INFORMS},
  doi       = {10.1287/moor.2015.0726},
}

@article{wang2011,
  author       = {Wang, Meng and Tan, Chee Wei and Xu, Weiyu and Tang, Ao},
  date         = {2011},
  doi          = {10.1109/TNET.2011.2150761},
  journaltitle = {IEEE/ACM Transactions on Networking},
  number       = {6},
  pages        = {1849--1859},
  title        = {Cost of Not Splitting in Routing: Characterization and Estimation},
  volume       = {19},
}

@article{HarksSchwarzPricing,
  author       = {Harks, Tobias and Schwarz, Julian},
  date         = {2023},
  doi          = {10.1137/21M1400924},
  journaltitle = {SIAM Journal on Optimization},
  number       = {2},
  pages        = {1223--1249},
  title        = {A Unified Framework for Pricing in Nonconvex Resource Allocation Games},
  volume       = {33},
}

@article{Daskalakis13,
  author       = {Daskalakis, Constantinos},
  date         = {2013},
  doi          = {10.1145/2483699.2483703},
  journaltitle = {ACM Transactions on Algorithms},
  number       = {3},
  pages        = {23:1--23:35},
  title        = {On the Complexity of Approximating a {Nash} Equilibrium},
  volume       = {9},
}

@article{ArrowDebreu,
  author       = {Arrow, Kenneth J. and Debreu, Gerard},
  date         = {1954},
  issn         = {0012-9682},
  journaltitle = {Econometrica},
  number       = {3},
  pages        = {265--290},
  title        = {Existence of an Equilibrium for a Competitive Economy},
  volume       = {22},
  doi          = {10.2307/1907353},
}

@article{Gruebel_et_al:2023a,
  author       = {Grübel, Julia and Huber, Olivier and Hümbs, Lukas and Klimm, Max and Schmidt, Martin and Schwartz, Alexandra},
  publisher    = {Taylor \& Francis},
  date         = {2023},
  doi          = {10.1080/10556788.2022.2117358},
  journaltitle = {Optimization Methods and Software},
  number       = {1},
  pages        = {153--183},
  title        = {Nonconvex equilibrium models for energy markets: exploiting price information to determine the existence of an equilibrium},
  volume       = {38},
}

@article{Liberopoulos_Andrianesis:2016,
  author       = {Liberopoulos, George and Andrianesis, Panagiotis},
  date         = {2016},
  doi          = {10.1287/opre.2015.1451},
  journaltitle = {Operations Research},
  number       = {1},
  pages        = {17--31},
  title        = {Critical Review of Pricing Schemes in Markets with Non-Convex Costs},
  volume       = {64},
}

@book{Rubinstein:1998,
  author    = {Rubinstein, Ariel},
  publisher = {The MIT Press},
  date      = {1998},
  title     = {Modeling Bounded Rationality},
  doi       = {10.7551/mitpress/4702.001.0001},
  isbn      = {9780262282758},
  location  = {Cambridge, Massachusetts},
  series    = {Zeuthen Lectures},
}

@incollection{Simon:1972,
  author    = {Simon, Herbert A.},
  publisher = {North-Holland},
  url       = {https://digitalcollections.library.cmu.edu/node/35364},
  booktitle = {Decision and Organization: A Volume in Honor of Jacob Marschak},
  date      = {1972},
  pages     = {161--176},
  title     = {Theories of Bounded Rationality},
  editor    = {McGuire, C. B. and Radner, Roy},
  location  = {Amsterdam},
  series    = {Studies in Mathematical and Managerial Economics},
  volume    = {12},
}

@article{sagratella2017computing,
  title     = {Computing equilibria of {Cournot} oligopoly models with mixed-integer quantities},
  author    = {Sagratella, Simone},
  journal   = {Mathematical Methods of Operations Research},
  volume    = {86},
  number    = {3},
  pages     = {549--565},
  year      = {2017},
  publisher = {Springer},
  doi       = {10.1007/s00186-017-0599-8},
}

@article{Deligkas20Continuous,
  author       = {Argyrios Deligkas and John Fearnley and Paul Spirakis},
  date         = {2020},
  journaltitle = {Algorithmica},
  title        = {{Lipschitz} Continuity and Approximate Equilibria},
  doi          = {10.1007/s00453-020-00709-3},
  pages        = {2927--2954},
  volume       = {82},
  number       = {10},
}

@article{Duguet25MNE,
  author  = {Aloïs Duguet and Margarida Carvalho and Gabriele Dragotto and Sandra Ulrich Ngueveu},
  title   = {Computing Approximate {Nash} Equilibria for Integer Programming Games},
  doi     = {10.1007/s11590-025-02221-5},
  journal = {Optimization Letters},
  year    = {2026},
  volume  = {20},
  number  = {1},
  pages   = {231--255},
}

@InProceedings{Hansknecht14,
  author    = {Hansknecht, Christoph and Klimm, Max and Skopalik, Alexander},
  booktitle = {Approximation, Randomization, and Combinatorial Optimization. Algorithms and Techniques (APPROX/RANDOM 2014)},
  title     = {{Approximate Pure {Nash} Equilibria in Weighted Congestion Games}},
  doi       = {10.4230/LIPIcs.APPROX-RANDOM.2014.242},
  editor    = {Jansen, Klaus and Rolim, Jos\'{e} and Devanur, Nikhil R. and Moore, Cristopher},
  isbn      = {978-3-939897-74-3},
  pages     = {242--257},
  publisher = {Schloss Dagstuhl -- Leibniz-Zentrum f{\"u}r Informatik},
  series    = {Leibniz International Proceedings in Informatics (LIPIcs)},
  volume    = {28},
  address   = {Dagstuhl, Germany},
  issn      = {1868-8969},
  urn       = {urn:nbn:de:0030-drops-47005},
  year      = {2014},
}

@article{Caragiannis15,
  author     = {Caragiannis, Ioannis and Fanelli, Angelo and Gravin, Nick and Skopalik, Alexander},
  title      = {Approximate Pure {Nash} Equilibria in Weighted Congestion Games: Existence, Efficient Computation, and Structure},
  doi        = {10.1145/2614687},
  issn       = {2167-8375},
  number     = {1},
  volume     = {3},
  address    = {New York, NY, USA},
  issue_date = {March 2015},
  journal    = {ACM Transactions on Economics and Computation},
  month      = {3},
  numpages   = {32},
  publisher  = {Association for Computing Machinery},
  year       = {2015},
  pages      = {2:1--2:32},
}

@article{Budish11Shapley,
  author  = {Budish, Eric},
  title   = {The Combinatorial Assignment Problem: Approximate Competitive Equilibrium from Equal Incomes},
  doi     = {10.1086/664613},
  number  = {6},
  pages   = {1061--1103},
  volume  = {119},
  journal = {Journal of Political Economy},
  year    = {2011},
}

@inproceedings{Deligkas24Inapproximability,
  author    = {Deligkas, Argyrios and Fearnley, John and Hollender, Alexandros and Melissourgos, Themistoklis},
  booktitle = {Proceedings of the 25th ACM Conference on Economics and Computation},
  title     = {Constant Inapproximability for {Fisher} Markets},
  doi       = {10.1145/3670865.3673533},
  isbn      = {9798400707049},
  location  = {New York, NY, USA},
  venue     = {New Haven, CT, USA},
  pages     = {13--39},
  publisher = {Association for Computing Machinery},
  series    = {EC '24},
  numpages  = {27},
  year      = {2024},
}

@inproceedings{Guruswami05envyfree,
  author    = {Guruswami, Venkatesan and Hartline, Jason D. and Karlin, Anna R. and Kempe, David and Kenyon, Claire and McSherry, Frank},
  booktitle = {Proceedings of the Sixteenth Annual ACM-SIAM Symposium on Discrete Algorithms},
  date      = {2005},
  title     = {On profit-maximizing envy-free pricing},
  pages     = {1164--1173},
  publisher = {Society for Industrial and Applied Mathematics},
  location  = {Philadelphia, Pennsylvania},
  series    = {SODA '05},
  isbn      = {9780898715859},
  url       = {https://dl.acm.org/citation.cfm?id=1070432.1070598},
}

@article{Vazirani11Market,
  author     = {Vazirani, Vijay V. and Yannakakis, Mihalis},
  title      = {Market equilibrium under separable, piecewise-linear, concave utilities},
  doi        = {10.1145/1970392.1970394},
  issn       = {0004-5411},
  number     = {3},
  volume     = {58},
  address    = {New York, NY, USA},
  issue_date = {May 2011},
  journal    = {Journal of the ACM},
  month      = {5},
  numpages   = {25},
  publisher  = {Association for Computing Machinery},
  year       = {2011},
  pages      = {10:1--10:25},
}

@inproceedings{Codenotti05Market,
  author    = {Codenotti, Bruno and McCune, Benton and Varadarajan, Kasturi},
  booktitle = {Proceedings of the Thirty-Seventh Annual ACM Symposium on Theory of Computing},
  title     = {Market equilibrium via the excess demand function},
  doi       = {10.1145/1060590.1060601},
  isbn      = {1581139608},
  location  = {New York, NY, USA},
  venue     = {Baltimore, MD, USA},
  pages     = {74--83},
  publisher = {Association for Computing Machinery},
  series    = {STOC '05},
  numpages  = {10},
  year      = {2005},
}

@inproceedings{Garg25Welfare,
  author    = {Jugal Garg and Yixin Tao and László A. Végh},
  booktitle = {Proceedings of the 2025 Annual ACM-SIAM Symposium on Discrete Algorithms (SODA)},
  date      = {2025},
  title     = {Approximating Competitive Equilibrium by {Nash} Welfare},
  doi       = {10.1137/1.9781611978322.83},
  pages     = {2538--2559},
  publisher = {Society for Industrial and Applied Mathematics},
  editor    = {Azar, Yossi and Panigrahi, Debmalya},
  isbn      = {9781611978322},
}

@book{Starr12Book,
  author    = {Starr, Ross M.},
  date      = {2011},
  title     = {General Equilibrium Theory: An Introduction},
  publisher = {Cambridge University Press},
  edition   = {2},
  doi       = {10.1017/CBO9780511975356},
  isbn      = {9780521826457},
  location  = {Cambridge},
}

@report{DSHS25,
  author      = {Duguet, Aloïs and Harks, Tobias and Schmidt, Martin and Schwarz, Julian},
  date        = {2025},
  title       = {\revms{Branch-and-Cut for Mixed-Integer {Nash} Equilibrium Problems}},
  doi         = {10.48550/arXiv.2506.02520},
  institution = {arXiv},
  type        = {Preprint},
  number      = {2506.02520},
}

@article{Liu23Shapley,
  author  = {Liu, Kang and Oudjane, Nadia and Wan, Cheng},
  title   = {Approximate {Nash} Equilibria in Large Nonconvex Aggregative Games},
  doi     = {10.1287/moor.2022.1321},
  number  = {3},
  pages   = {1791--1809},
  volume  = {48},
  journal = {Mathematics of Operations Research},
  year    = {2023},
}

\end{document}